\RequirePackage{fix-cm}
\documentclass[envcountsect,numbook]{svjour3_arxiv4b}   
\smartqed  

\usepackage{graphicx}
\usepackage[utf8]{inputenc}
\usepackage{dsfont}
\usepackage{subfigure}
\usepackage{xcolor}
\usepackage{url}
\usepackage{enumerate} 
\usepackage{booktabs}
\usepackage{tabularx}
\usepackage{ltablex}
\usepackage{mathptmx}
\usepackage{amssymb,amsmath,amsfonts,latexsym} 
\usepackage{longtable} 
\usepackage{textcomp}
\usepackage{lscape}
\usepackage[12pt]{extsizes}
\usepackage{geometry}
\usepackage{ifthen}
\usepackage[pdftex]{hyperref}
\usepackage[normalem]{ulem}
\usepackage{mathtools}
\mathtoolsset{showonlyrefs}
\usepackage{enumitem}

\DeclareMathAlphabet{\pazocal}{OMS}{zplm}{m}{n}
\let\mathcal\pazocal
\allowdisplaybreaks

\spnewtheorem{theorem}{Theorem}{\bfseries}{\itshape}
\spnewtheorem{cor}[theorem]{Corollary}{\bfseries}{\itshape}
\spnewtheorem{lemma}[theorem]{Lemma}{\bfseries}{\itshape}
\spnewtheorem{proposition}[theorem]{Proposition}{\bfseries}{\itshape}
\spnewtheorem{definition}[theorem]{Definition}{\bfseries}{\itshape}
\spnewtheorem{remark}[theorem]{Remark}{\bfseries}{\upshape}
\spnewtheorem{assumption}[theorem]{Assumption}{\bfseries}{\itshape}
\counterwithin{theorem}{section}
\smartqed

\renewcommand{\paragraph}[1]{{\bf #1.}}
\definecolor{myred}{rgb}{0.8,0,0}  

\newcommand{\mymarginpar}[1]{ \marginpar{{\tiny #1}}}

\renewcommand{\mymarginpar}[1]{}
{\vskip\baselineskip\noindent\textbf{Proof of {#1}:}}%
{\hspace*{.1pt}\hspace*{\fill}\BOX\vskip\baselineskip}

\renewcommand{\paragraph}[1]{{\bf #1.}}
\definecolor{myred}{rgb}{0.8,0,0}  
\definecolor{mygreen}{rgb}{0,0.7,0}
\definecolor{mygrey}{rgb}{0.5,0.5,0.5}

\def \R{\mathbb{R}}               
\def \N{\mathbb{N}}               
\def \P{\mathbb{P}}             
\def \1{{\bf 1}}                
\def \0{{\bf 0}}
\def\utility{\mathcal U}

\def\qed{\hfill$\Box$}

\def\eps{\varepsilon}

\def\revlevel{\overline{\mu}}
\def\revspeed{\kappa}
\def\contexp{\zeta}
\def\driftinitial{\overline{\filter}_0}
\def\filterinitial{\filter_0}
\def\covinitial{\overline{\variance}_0}
\def\condcovinitial{\variance_0}

\def\Abound{A_{\cpsi}}
\def\Bbound{B_{\cpsi}}
\def\Cbound{C_{\cpsi}}
\def\para{\varrho}
\def\modellparameter{\mathcal P}

\def\cpsi{\gamma}    
\def\Sigmafac{P}    

\def\filter{m}
\def\variance{q}

\def\HG{G}
\def\HF{F}
\def\HR{R}
\def\HD{J}
\def \myzeta{J}
\def\contexp{J}
\def\HC{Z}
\def\nAktien{d}

\def\nWienerRendite{d_1}
\def\nWienerDrift{d_2}
\def\nWienerExperten{d_3}

\def\varianceexp{\Gamma}

\def\volR{\sigma_R}
\def\voldrift{\sigma_{\mu}}
\def\volexp{\sigma_{\myzeta}}
\def\alphamm{\alpha}
\def\betam{\beta}
\def\gammam{\gamma_{\Mpro}}

\def\alphaqq{\alpha_{\Qpro}}

\def\gammaq{\gamma_{\Qpro}}
\def\welth{X}
\def\drift{\mu}

\def\komppoi{\widetilde{N}}

\def\valuefkt{V}
\def\valueorigin{\mathcal V}

\def\rewardorigin{\mathcal D}
\def\rewardconstant{C_{V}}
\def\rewardstate{D}

\def\wealth{X}

\def\Mpro{M}
\def\Qpro{Q}

\def\qed{\hfill$\Box$}

\def\varexp{\Gamma}

\def \trace{\operatorname{tr}}
\def \E{{\mathbb{E}}}

\def \diag{\operatorname{diag}}

\def\Dfun{g}
\def\myqed{\qed}
\def\quadform{L}
\def\stateall{L}
\def\statespaceall{\mathcal{L}}

\begin{document}
	
	\title{Well Posedness of Utility Maximization Problems Under Partial Information in a Market with Gaussian Drift
	}
	\titlerunning{Well Posedness of Utility Maximization Problems Under Partial Information}
	\author{Abdelali Gabih \and Hakam Kondakji \and Ralf Wunderlich}
	
	\authorrunning{A.~Gabih,  H.~Kondakji, R.~Wunderlich} 
	
	\institute{Abdelali Gabih \at
		Equipe de Modélisation et Contrôle des Systèmes Stochastiques et Déterministes, 
		Faculty of Sciences,  Chouaib Doukkali University, El Jadida 24000, Morocco,  \email{\texttt{a.gabih@uca.ma}}
		\and	
		Hakam Kondakji \at {Westfälische Hochschule, Fachbereich Wirtschaft und Informationstechnik, Münsterstr. 265,
			46397 Bocholt, Germany,	}
		\email{\texttt{hakam.kondakji@w-hs.de}}
		\and 
		Ralf Wunderlich \at
		Brandenburg University of Technology Cottbus-Senftenberg, Institute of Mathematics, P.O. Box 101344, 03013 Cottbus, Germany,
		\email{\texttt{ralf.wunderlich@b-tu.de}}
	}

	\date{Version of \today}
	\maketitle
	
	\begin{abstract}    
		This paper investigates well posedness of utility maximization problems for financial markets where stock returns depend on a hidden Gaussian mean-reverting drift process. Since that process is potentially unbounded, well posedness cannot be guaranteed  for utility functions which are not bounded from above. For power utility with relative risk aversion smaller than that of log-utility this leads to restrictions on the choice of model parameters such as the investment horizon and parameters controlling the variance of the asset price and drift processes. We derive sufficient conditions to the model parameters leading to bounded  maximum expected utility of terminal wealth for models with full and partial information. 
	\end{abstract}

	\section{Introduction}
	\label{introduction}
	
	In this paper we investigate utility maximization problems  for a
	financial market where asset prices follow a diffusion process with
	an unobservable Gaussian mean-reverting drift  modelled by an
	Ornstein-Uhlenbeck process. It is  motivated by our papers \cite{Gabih et al (2022) PowerFixed,Gabih et al (2022) PowerRandom} where we examine in detail the maximization of expected power utility of terminal wealth which is treated as a stochastic optimal control problem  under partial information. A special feature of these papers is that for the construction of  optimal portfolio strategies investors estimate the unknown drift not only from observed asset prices. They  also  incorporate external sources of information such as news, company reports, ratings or their own intuitive views on the
	future asset performance.   These outside sources of information  are called
	expert opinions. 
	
	In the present paper, we focus on the well posedness of the above stochastic control problem,   which is often overlooked or taken for granted  in the literature.  For Gaussian drift processes which are potentially unbounded, well posedness in general cannot be guaranteed  for utility functions which are not bounded from above.
	This is the case for log-utility and power utility with relative risk aversion smaller than that of log-utility. For log-utility well posedness can be shown quite easily and holds without restriction to the model parameters. However, the case of power utility is much more demanding and leads to restrictions on the choice of model parameters such as the investment horizon, the risk aversion  parameter of the power utility function, parameters controlling the variance of the asset price and drift processes,  and the filter process describing the conditional covariance of the Kalman filter.	
	
	\paragraph{Literature review}
	The above phenomenon was already observed in Kim and
	Omberg \cite{Kim and Omberg (1996)} for a financial market with an observable drift modeled by an Ornstein-Uhlenbeck process. They  coined the terminology \textit{nirvana strategies}. Such strategies generate in finite time  a
	terminal wealth with a distribution leading to infinite expected utility. Note that this is a property of the distribution of terminal wealth  and realizations of terminal wealth   need not to be infinite. The same holds for  the generating strategies which  might be even suboptimal.    That phenomenon was also observed in Korn and Kraft \cite[Sec. 3]{Korn and Kraft (2004)} who coined it ``I-unstability'', in  Angoshtari \cite{Angoshtari2013,Angoshtari2016}  and Lee and Papanicolaou  \cite{Lee Papanicolaou (2016)} who studied  power utility maximization problems and their well posedness for financial market models with  cointegrated  asset price processes and in Battauz et al.~\cite{Battauz et al (2017)} for markets with defaultbale assets. For the case of partial information  Colaneri et al.~\cite{Colaneri et al (2021)} provides some results   for markets with a single risky asset ($d=1$). Kim and Omberg \cite{Kim and Omberg (1996)} also pointed out that financial market models allowing investors to attain nirvana do not properly reflect reality. Thus, there are not only mathematical reasons to exclude combinations of model parameters allowing for attaining nirvana, i.e., not ensuring well-posed optimization problems. This problem is addressed in the present  paper and we derive sufficient conditions to the model parameters leading to bounded  maximum expected utility of terminal wealth for portfolio selection problems under full and partial information.

	\smallskip
	Portfolio selection problems for market models with partial information on the drift have been intensively studied in the last years. We refer to Lakner \cite{Lakner (1998)} and  Brendle \cite{Brendle2006} for models with Gaussian drift, and to Rieder and Bäuerle \cite{Rieder_Baeuerle2005}, Sass and Haussmann \cite{Sass and Haussmann (2004)} for models in which the drift is described by a continuous-time hidden Markov chain. A generalization of these approaches and further references can be found  in Björk et al.
	\cite{Bjoerk et al (2010)}. 
	
	Utility maximization problems for investors with logarithmic preferences in market models with non-observable Gaussian drift process   and discrete-time expert opinions are addressed in a series of papers \cite{Gabih et al (2014),Gabih et al (2019) FullInfo,Sass et al (2017),Sass et al (2021),Sass et al (2022)} of the present authors and of Sass and Westphal. The case of continuous-time expert opinions and power utility maximization is treated in a series of papers by Davis and Lleo, see \cite{Davis and Lleo (2013_1),Davis and Lleo (2020),Davis and Lleo (2022)}. For models with drift processes described by  continuous-time hidden Markov chains and power utility maximization we refer to  Frey et al.~\cite{Frey et al. (2012),Frey-Wunderlich-2014}. Finally, the computation of optimal strategies using dynamic programming methods for the   power utility maximization problems addressed in this paper can be found in our  papers \cite{Gabih et al (2022) PowerFixed,Gabih et al (2022) PowerRandom}. 
	
	\paragraph{Our contribution}
	The paper addresses well posedness of power utility maximization problems under partial information on the not directly observable drift of risky assets.  It derives sufficient conditions to the model parameters ensuring bounded objective functions, and  under which the dynamic programming approach can be applied for their solution. Such conditions are often taken for granted or overlooked and restrict the choice of model parameters for investors which are less risk averse than the	log-utility investors.

	To the best of our knowledge, our results for the case of multi-asset markets and partial information are new to the literature. They extend known results for the corresponding optimization problems under full information to the case of partial information. A  first main result is Theorem \ref{theo_bound_V}   providing an upper bound for the  expected  utility of terminal wealth expressed in terms of the solution to some matrix Riccati differential equation, and involving the  current value of the non-observable drift.  At first glance this approach does not seem very practical, but there is a rich specialist literature on properties of solutions  Riccati equations, see  Kucera \cite{Kucera (1973)}, Roduner \cite{Roduner (1994)}, and Wonham \cite{Wonham (1968)}, and  the references  therein, and nowadays numerical solutions  are standard. 	
	The  result of Theorem \ref{theo_bound_V} allows to 
	deduce sufficient conditions to the model parameters ensuring the well posedness of the utility maximization problem under full information in Corollary \ref{suff_cond_full}.
	The respective conditions for the case of partial information follow from the  projection of the above upper bound  on the  investor filtration of the partially informed investor an lead to our second main result given in Theorem \ref{theorem_partial_Inv}. It allows to derive well posedness conditions for the problem under partial information Corollary \ref{suff_cond_partial}.  We also provide numerical results to illustrate the theoretical findings for a market model with a single risky asset. Here,  the sufficient conditions for well posedness  become quite explicit. This allows an insightful visualization of the set of feasible model parameters.

	Note that the actual solution of the optimization problems is not intended but already studied, e.g., in our  papers \cite{Gabih et al (2022) PowerFixed,Gabih et al (2022) PowerRandom}. 	The derived results appear to be helpful for the analysis of portfolio selection problems under partial information in general, and not only to the specific situation where investors draw information for estimating unobservable drifts from return observations which are combined with additional information from expert opinions, as in this work.

	\smallskip
	{\paragraph{Paper organization}}
	In Section \ref{market_model} we introduce the financial market model with partial  information on the drift and formulate the portfolio optimization problem. The well posedness of that problem is studied in Section \ref{sec_wellposedness}. We derive sufficient conditions to the model parameters ensuring the well posedness of the utility maximization problem under full as well as partial information. 
	These conditions become quite explicit for market models with a single risky asset which are considered in Subsection \ref{WellPosedSpecialCase}.  Section \ref{numeric_result} illustrates the theoretical findings by results of some numerical experiments and  visualizes the derived restrictions on the model parameters. The appendix collects proofs which are removed from the main text.

	\smallskip
	\paragraph{Notation} Throughout this paper, we use the notation $I_d$ for the identity matrix in $\R^{d\times d}$,  $0_{d}$
	denotes the null vector in $\R^d$, $0_{d\times m}$ the null matrix
	in $\R^{d\times m}$. For a symmetric and positive-semidefinite
	matrix $A\in\R^{d\times d}$ we call a symmetric and
	positive-semidefinite matrix $B\in\R^{d\times d}$ the \emph{square
		root} of $A$ if { $B^2=A$}. The square root is unique and will be
	denoted by $A^{1/2}$. For a generic process $X$ we denote by
	$\mathbb{G}^X$ the filtration generated by $X$.
	
	\section{Financial Market and Optimization Problem}
	\medskip
	
	\label{market_model}
	\subsection{Price Dynamics}
	\label{PriceDynamics}   Our financial market model comprises one
	risk-free and multiple risky assets. The  setting is based on Gabih
	et al.~\cite{Gabih et al (2014),Gabih et al (2019) FullInfo} and
	Sass et al.~\cite{Sass et al (2017),Sass et al (2022),Sass et al
		(2021)}  and also used in our   papers \cite{Gabih et al (2022) PowerFixed,Gabih et al (2022) PowerRandom}. For a  fixed date $T>0$ representing the investment
	horizon, we work on a filtered probability space
	$(\Omega,\mathcal{G},\mathbb{G},\P)$, with filtration
	$\mathbb{G}=(\mathcal {G}_t)_{t \in [0,T]}$ satisfying the usual
	conditions. All processes are assumed to be $\mathbb{G}$-adapted.

	We consider  in our  market model   discounted asset prices with  the risk-free asset as numéraire. Then the  risk-free asset has a 	constant price $S^0_t=1$. Further, there are  $\nAktien$ risky securities whose excess log-return or risk premium process
	$R=(R^{1},\ldots,R^{\nAktien})$ is defined by the SDE
	\begin{align}
		dR_t=\mu_t\, dt+\volR\, dW^{R}_t, \label{ReturnPro}
	\end{align}
	for a given $\nWienerRendite$-dimensional $\mathbb{G}$-adapted
	Brownian motion $W^{\HR}$ with  $\nWienerRendite\geq\nAktien$. The
	volatility matrix $\volR\in\mathbb
	R^{\nAktien\times\nWienerRendite}$ is  assumed to be constant over
	time such that $\Sigma_{R}:=\volR\volR^{\top}$ is positive definite.  In the remainder of this paper we will call $R$ simply \textit{returns}.
	In this setting the discounted price process $S=(S^1,\ldots,S^{\nAktien})$ of
	the risky securities reads as
	\begin{align}
		dS_t&=diag(S_t)\, dR_t,~~ S_0=s_0, \label{stockmodel}
	\end{align}
	with some fixed initial value $s_0=(s_0^1,\ldots,s_0^d)$. Note that
	for the solution to the above SDE it holds
	\begin{align*}
		\log S_t^{i}-\log s_0^{i} &= \int\limits_0^t \drift_s^{i}ds
		+\sum\limits_{j=1}^{\nWienerRendite}\Big(
		\sigma_R^{ij}W_t^{R,j}-\frac{1}{2} (\sigma_R^{ij})^2 t\Big)
		=R_t^{i}-\frac{1}{2}\sum\limits_{j=1}^{\nWienerRendite}
		(\sigma_R^{ij})^2 t ,\quad i=1,\ldots,\nAktien.
	\end{align*}
	So we have the equality $\mathbb{G}^R = \mathbb{G}^{\log S} =
	\mathbb{G}^S$. 
	
	The dynamics of the drift process $\mu=(\mu_t)_{t\in[0,T]}$ in \eqref{ReturnPro}
	are
	given by the stochastic differential equation (SDE)  defining an Ornstein--Uhlenbeck process
	\begin{align}
		\label{drift} d\mu_t=\revspeed(\revlevel-\mu_t)\, dt+\voldrift\,
		dW^{\mu}_t,
	\end{align}
	where $\revspeed\in\mathbb R^{\nAktien\times\nAktien}$,
	$\voldrift\in\mathbb R^{\nAktien\times\nWienerDrift}$ and
	$\revlevel\in\mathbb R^{\nAktien} $ are constants. We require that   all eigenvalues of  $\revspeed$ have a positive real part (that is, $-\kappa$ is a stable matrix)  and that $\Sigma_{\mu}:=\voldrift\voldrift^{\top}$
	is positive definite. Further, $W^{\mu}$ is a
	$\nWienerDrift$-dimensional Brownian motion such that
	$\nWienerDrift\geq\nAktien$.   For the sake of simplification and 	shorter notation, we assume  that the Wiener processes $W^{R}$ and 	$W^{\mu}$ driving the return and drift process, respectively, are 	independent.  For the general case we refer to Brendle \cite{Brendle2006},  Colaneri et
	al. \cite{Colaneri et al (2021)} and  Fouque et al.~\cite{Fouque et
		al. (2015)}. Here, $\revlevel$ is the
	mean-reversion level, $\revspeed$ the mean-reversion speed and
	$\voldrift$ describes the volatility of $\mu$. The initial value
	$\drift_0$ is assumed to be a normally distributed random variable
	independent of $W^{\mu}$ and $W^{R}$ with mean $\driftinitial\in
	\R^{\nAktien}$ and covariance matrix $\covinitial\in\mathbb
	R^{\nAktien\times\nAktien}$ assumed to
	be symmetric and  positive semi-definite. 

	\subsection{Partial Information}\label{Expert_Opinions}  
	Our mathematical market model reflects the fact that investors in real financial markets do not have full access to market information.
	They can instead observe the historical data of the return process $R$,	 but they neither observe the factor process $\mu$ nor
	the Brownian motion $W^{R}$. Further, investors know the	model parameters	such as $\volR,\revspeed, \revlevel, \voldrift $  and
	the distribution $\mathcal{N}(\driftinitial,\covinitial)$ of  the
	initial value $\drift_0$. 
	
	Information about the drift $\mu$ can be
	drawn from observing the  asset prices from which the returns $R$ can be derived.  However,	
	it is well-known that estimating the drift  with a reasonable degree of precision based only on historical asset prices is known to be almost impossible. This is nicely described in   Rogers \cite[Chapter 4.2]{Rogers (2013)}. Here, the author analyzes that problem  for a model in which the drift is even constant. Reliable estimate require extremely long time series of data  which are usually not available. Furthermore, the assumption of a constant drift over longer periods of time is rather unrealistic. Drifts tend to fluctuate randomly over time and drift effects are often overlaid by volatility.

	For	these reasons, portfolio
	managers and traders also rely on external sources of information such as news, company reports, ratings and benchmark values. Further, they increasingly turn to  data outside of the traditional sources that companies and financial markets provide. Examples are social media posts, internet search, satellite imagery, sentiment indices, pandemic data,  product review trends and are often related to Big Data analytics.
	
	In the literature, these external sources of information are referred to as \textit{expert opinions} or more generally as \textit{alternative data}, see Chen and Wong \cite{Chen Wong (2022)}, Davis and Lleo \cite{Davis and Lleo (2022)}. We use the first term here. After appropriate mathematical modeling, they are included as additional noisy observations in the drift estimation and the construction of optimal portfolio strategies.	This approach goes back to Black and	Litterman~\cite{Black_Litterman (1992)} and their celebrated Black-Litterman model, which is an extension of the classic one-period Markowitz model.

	A first modeling approach considers expert opinions  as noisy signals about
	the current state of the drift arriving at discrete time points forming an increasing sequence $(T_k)_{k\in \mathbb I}$  with values in $[0,T]$  and $\mathbb I\subseteq\N_0$. The literature distinguishes between a given finite number of deterministic time points as in \cite{Gabih et al (2014),Gabih et al (2022) PowerFixed,Sass et al (2017),Sass et al (2022)} or random time points which are the jump times of a Poisson process with some given intensity as in \cite{Gabih et al (2019) FullInfo,Sass et al (2022)}.    
	The signals or ``expert views'' at time $T_k$ are modelled by
	$\R^\nAktien$-valued  Gaussian random vectors
	$Z_k=(Z_k^1,\ldots,Z_k^{\nAktien})^{\top}$ with
	\begin{align}
		\label{Expertenmeinungen_fest}
		Z_k=\drift_{T_k}+{\varianceexp}^{\frac{1}{2}}\varepsilon_k,\quad
		k\in\N_0,
	\end{align}
	where the matrix  $\varianceexp\in\R^{\nAktien\times\nAktien}$ is
	symmetric and positive definite.
	Further,  $(\varepsilon_k)_{\in\N_0}$ is a sequence of independent
	standard normally distributed random vectors, i.e.,
	$\varepsilon_k\sim \mathcal{N}(0,I_d)$. 
	It is  also independent of
	both the Brownian motions $W^R, W^\mu$ and the initial value $\mu_0$
	of the drift process. Thus given $\mu_{T_k}$ the expert
	opinion $Z_k$ is $\mathcal{N}(\mu_{T_k},\varianceexp)$-distributed.
	So, $Z_k$ can be considered as an unbiased estimate of the unknown
	state of the drift at time $T_k$. The matrix $\varianceexp$ is a
	measure of the expert's reliability.  Its diagonal entries $\varexp^{ii}$
	are just the variances of the expert's
	estimate of the drift component $\drift^i$ at time $T_k$.  The larger $\varexp^{ii}$, the less
	reliable the expert's estimate is.

	Expert opinions can also take the form of relative views, which are estimates of the difference in drift between two stocks rather than an absolute view of the drift of a single stock. We refer for this extension  to 	Sch\"ottle et al.~\cite{Schoettle et al. (2010)} where the authors
	show how to switch between these two models for expert opinions by	means of a pick matrix.
	
	
	In a second modeling approach the expert opinions do not arrive  at discrete time points  but continuously over time as in
	the BLCT model of Davis and Lleo \cite{Davis and Lleo
		(2013_1),Davis and Lleo (2020),Davis and Lleo (2022)}. This is motivated by the results of
	Sass et al. who show in \cite{Sass et al (2022)} for $n\in\N$ periodically arriving expert opinions that  for increasing number $n$ and an  expert's variance $\Gamma$
	growing linearly in $n$, asymptotically for $n\to \infty$   the information drawn from these expert opinions is  essentially
	the same as the information one gets from observing yet another
	diffusion process. A similar result is obtained in    \cite{Sass et al (2021)} for randomly arriving expert opinions with increasing 	arrival intensity $\lambda$ and an  expert's variance $\Gamma$	growing linearly $\lambda$.	
	This above mentioned diffusion process can then be interpreted as
	an expert who gives a continuous-time estimation about the current
	state of the drift. Let this estimate be given by the diffusion
	process
	\begin{align}\label{continuous-expert}
		d\contexp_t = \mu_t dt +\volexp dW_t^{\contexp},
	\end{align}
	where $W_t^{\contexp}$ is a $\nWienerExperten$-dimensional Brownian
	motion independent of $W_t^R$ and $W^{\mu}$ and such that with
	$\nWienerExperten\geq\nAktien$. The volatility   matrix
	$\volexp\in\mathbb R^{\nAktien\times\nWienerExperten}$ is assumed to
	be constant over time such that the matrix
	$\Sigma_{\contexp}:=\volexp\volexp^{\top}$ is positive definite.
	In \cite{Gabih et al (2022) PowerFixed}, we show that, based on
	this model and on the diffusion approximations provided in
	\cite{Sass et al (2022),Sass et al (2021)},    efficient
	approximative solutions can be found to utility maximization problems for
	partially informed investors observing high-frequency discrete-time
	expert opinions.
	
	\subsection{Investor Filtration}
	\label{Investor_Filtration}  
	In view of the different levels of information on the financial markets, we consider different types of investors. Mathematically, the information available to an investor can be described by the \textit{investor filtration} 	$\mathbb{F}^H=(\mathcal{F}^H_t)_{t\in[0,T]}$. Here, $H$ denotes the 	information regime for which we
	consider the cases $H=\HR, Z_n,Z_\lambda ,\HD,\HF$  and the investor with filtration 	$\mathbb{F}^H=(\mathcal{F}^H_t)_{t\in[0,T]}$ is called the $H$-investor.
	The $\HR$-investor only observes the return process $R$.  The $Z_n$-investor combines	observations of returns with the discrete-time expert opinions
	$Z_k$  arriving at $n\in \N$ known and fixed arrival times  $ 0\le t_0<\ldots <t_{n-1}  <T$, while the $Z_\lambda$-investor combines returns with   expert opinions
	$Z_k$  arriving at the random jump times of an homogeneous Poisson process with intensity $\lambda>0$.	
	The $\HD$-investor observes the return process together with the continuous-time expert $\HD$. Finally, the $\HF$ investor is fully informed and can directly observe the drift process $\mu$ and, of course, the return process.
	In a market with stochastic drift, this case is not realistic, but we use it as a benchmark.

	The $\sigma$-algebras $\mathcal{F}^H_t$ representing the  $H$-investor's information  at time $t\in[0,T]$ are defined at initial time  $t=0$ by $\mathcal{F}^\HF_0=\sigma\{\drift_0\}$ for the fully informed investor,    and by  $\mathcal{F}^H_0=\mathcal{F}^I_0\subset \mathcal F_0^{\HF}$ for $H=\HR, Z_n,Z_\lambda ,\HD$, i.e., for the partially informed investors. Here, $\mathcal{F}_0^I$ denotes the $\sigma$-algebra representing   prior information about the initial drift $\mu_0$,   e.g., from
	observing  returns  or expert opinions in the past, before the
	trading period $[0,T]$. Note that all partially informed investors ($H=\HR, Z_n,Z_\lambda ,\HD$) start  at $t=0$ with the same initial information given by
	$\mathcal{F}_0^I$. For $t\in (0,T]$ we define 
	\[\begin{array}{rl}
		\mathcal {F}_t^{\HR}&=\sigma(R_s,~ s\le t)  \vee \mathcal{F}_0^I, \\[0.5ex]
		\mathcal {F}_t^{Z_n}&= \sigma(R_s, s\le t,\,~(Z_k),~ {t_k}\le t) \vee \mathcal{F}_0^I, \\[0.5ex]
		\mathcal {F}_t^{Z_\lambda}&=\sigma(R_s, s\le t,\,~({T_k},Z_k),~ {T_k}\le t) \vee \mathcal{F}_0^I, \\[0.5ex]  
		\mathcal {F}_t^{\HD}&=\sigma(R_s,\contexp_s,~ s\le t) \vee \mathcal{F}_0^I, \\[0.5ex]
		\mathcal {F}_t^{F}&=\sigma(R_s, \mu_s,~ s\le t).
	\end{array}
	\]
	It is assumed that the above $\sigma$-algebras $\mathcal{F}_t^H$	are augmented by the null sets 
	of $\P$.  {Further, we} assume that the conditional distribution
	of the initial value drift $\mu_0$ given $\mathcal{F}_0^{ I}$ is the
	normal distribution $\mathcal{N}(m_0,q_0)$ with mean
	$\filterinitial\in \R^{\nAktien}$ and covariance matrix
	$\condcovinitial\in\mathbb R^{\nAktien\times\nAktien}$ assumed to be
	symmetric and  positive semi-definite.  For more details and examples we refer to \cite[Sec. 2.3]{Gabih et al (2022) PowerFixed}.

	\subsection{Drift Estimates and Filtering}
	\label{Filtering}  
	The investors' trading decisions are based on their knowledge of the drift process $\mu$. The fully informed 
	$F$-investor observes the drift directly, the partially informed $H$-investor for
	$H=R, Z_n,Z_\lambda ,\HD$ must estimate the drift. This is a filtering problem 
	with hidden signal process $\mu$ and observations given by
	the returns $R$  for $H=\HR$, the returns combined with the sequence of  expert opinions  $(Z_k)$ for $H=Z_n$, or the double sequence of arrival times and expert opinions 
	$(T_k,Z_k)$ for $H=Z_\lambda$, or the continuous-time expert  $J$ for $H=\HD$. The \textit{filter} for the drift  $\mu_t$ is its
	projection on the $\mathcal{F}_t^H$-measurable random variables
	described by the conditional distribution of the drift given
	$\mathcal{F}_t^H$. The mean-square optimal estimator for the drift
	at time $t$, given the available information is  known to be  the
	\textit{conditional mean}  
	$$\Mpro_t^{H}:=\E[\mu_t\mid\mathcal{F}_t^H].$$
	The estimator's accuracy  can be described by the
	\textit{conditional covariance matrix}
	\begin{align}\nonumber
		\Qpro_t^{H}:=\E[(\mu_t-\Mpro^{H}_t)(\mu_t-\Mpro^{H}_t)^{\top}\mid\mathcal{F}^{H}_t].
	\end{align}
	In our market model  the signal   $\mu$, the observations
	and the initial value of the filter  are jointly Gaussian. Therefore,  we are in the setting of the Kalman filter, and  the
	conditional distribution of the drift at time $t$  is Gaussian, which is completely characterized by the
	conditional mean $\Mpro_t^{H}$ and the conditional covariance	$\Qpro_t^{H}$. For the associated filter equations describing the dynamics of the filter processes $\Mpro^{H}$ and $\Qpro^{H}$ we refer to \cite{Gabih et al (2022) PowerFixed,Sass et al (2017),Sass et al (2022)}  for the case of expert opinions arriving at fixed arrival times and to  \cite{Gabih et al (2019) FullInfo,Gabih et al (2022) PowerRandom,Sass et al (2022)} for random arrival times which are the jump times of a Poisson process. While $\Mpro^H$ solves a SDE driven by the return process with random jumps at the expert's arrival dates, the conditional covariance $\Qpro^H$ is governed by a Riccati ODE between the arrival dates and exhibits jumps at the arrival dates. The jump sizes are a deterministic function of $\Qpro^H$  before the jump.  Thus,  for fixed arrival times,  $\Qpro^H$ is deterministic and can be computed offline already in advance whereas for random arrival times,   $\Qpro^H$ is only piecewise deterministic. Then it has to be computed online and to be included into the state of associate control problems.

	\subsection{Portfolio and Optimization Problem}
	The self-financing trading of an investor can be described by the initial
	capital $x_0>0$ and the  $\mathbb{F}^H$-adapted trading strategy
	$\pi=(\pi_t)_{t\in[0,T]} $ with $\pi_t\in\R^{\nAktien}$. The $i$-th component 
	$\pi_t^{i}$ represents the proportion of the current portfolio wealth invested in the
	$i$-th stock at time $t$.  The assumption that $\pi$ is
	$\mathbb{F}^H$-adapted reflects that investment decisions have to be
	based only on information available to the $H$-investor.  These are observations 
	of assets returns for $H=R$,  returns combined with expert
	opinions for $H= Z_n,Z_\lambda ,\HD$, or returns combined with the drift process for $H=F$. Following
	the strategy $\pi$   the investor generates a wealth process
	$(X_t^{\pi})_{t\in [0,T]}$ whose dynamics  reads as
	\begin{align} \label{wealth_phys}
		\frac{dX_t^{\pi}}{X_t^{\pi}}= \pi_t^{\top}dR_t &= 
		\pi_t^{\top}\mu_t\; dt+\pi_t^{\top}\volR\; dW_t^{R},\quad
		X_0^{\pi}=x_0.
	\end{align}
	We denote by
	\begin{equation}
		\label{set_admiss_0} \mathcal{A}^H=\Big\{\pi= (\pi_t)_{t}  \colon
		\pi_t\in\mathbb R^{\nAktien}, \text{ $\pi$ is $\mathbb{F}^H$-adapted
		}, X^\pi_t > 0,\, \E\Big[ \int\nolimits_0^T \|\pi_t\|^2\, dt
		\Big]<\infty \Big\}
	\end{equation}%
	the  class of {\em admissible trading strategies}. The investor aims to maximize  expected utility of terminal
	wealth using a utility function $ \utility : \R_+\rightarrow\R$ which 
	models the risk aversion of the investor. Here, we 	use  the power utility function
	\begin{align}
		\label{util_def}
		\utility_{\theta}(x):=\frac{x^{\theta}}{\theta},\quad
		\theta\in(-\infty,0)\cup(0,1).
	\end{align}
	As limiting case for $\theta\rightarrow 0$ the family of power utility
	function contains the logarithmic utility $\utility_0(x):= \log  x$, 
	since we have
	$\utility_{\theta}(x)-\frac{1}{\theta}=\frac{x^{\theta}-1}{\theta}
	\xrightarrow[~\theta\rightarrow 0~]{\text{}} \log x.$
	The optimization problem thus reads as
	\begin{align}
		\label{opti_org} 		 
		\mathcal V_0^H:=\sup\limits_{\pi\in\mathcal{A}^{H}}	\rewardorigin_0^H(\pi) \quad \text{with}\quad
		\rewardorigin_0^H(\pi) =	\E\left[\utility_{\theta}(X_T^{\pi})~|~\mathcal{F}^H_0 	\right], \quad \pi\in\mathcal A^{H},
	\end{align}
	where we call $\rewardorigin_0^H(\pi)$ \textit{reward}
	or \textit{performance} of the strategy $\pi$ and $
	\mathcal V_0^H$ the \textit{value} of the problem to given model parameters, in particular to given initial
	capital $x_0$.
	For $H\neq F$ this is an optimization   problem under partial
	information since we have required that the strategy $\pi$ is
	adapted to the investor filtration $\mathbb F^{H}$. However, the drift
	coefficient of the wealth equation \eqref{wealth_phys} contains the non-observable drift $\mu$ and is therefore 
	not	$\mathbb F^{H}$-adapted. For $x_0 > 0$ the solution of the SDE
	\eqref{wealth_phys} is strictly positive. This ensures that the terminal wealth 
	$X_T^{\pi}$ is in the domain of logarithmic and power utility.
	
	For problems of the above type, in the literature as outlined in Sec.\ref{introduction}, dynamic programming is a powerful solution method which is frequently applied. The key idea is to embed the optimization problem \eqref{opti_org} into a family of problems in which the initial date is moved from $t=0$ to an arbitrary time point $t\in[0,T]$, and the initial value of the wealth process  $X_0^\pi=x_0$, as well as those of other state processes included in the analysis, are replaced by the respective values of the states at time $t$. Then one ties all these problems together and derives a partial differential equation known as the Hamilton-Jacobi-Bellman (HJB) equation. 	
	
	We introduce for a fixed strategy $\pi\in\mathcal A^{H}$ the notation $\mathcal{F}^{H,X}_t=\mathcal{F}^H_t \vee \sigma\{X_t^\pi\}$ and note that $\mathcal{F}^{H,X}_0=\mathcal{F}^H_0$ since $X_0^\pi=x_0$ is the given and fixed initial capital. Then the  optimization problems of the above mentioned  family are indexed by   time $t\in[0,T]$ and read
	\begin{align}
		\label{opti_t} \mathcal
		V_t^H:=\sup\limits_{\pi\in\mathcal{A}^{H}}	\rewardorigin_t^H(\pi) \quad  \text{with} \quad 
		\rewardorigin_t^H(\pi) =
		\E\left[\utility_{\theta}(X_T^{\pi})~|~ \mathcal{F}^{H,X}_t  \right],\quad \pi\in\mathcal A^{H}.
	\end{align}		
	To follow the dynamic programming approach, for the information regimes $H=R, Z_n,Z_\lambda ,J$ with partial information, an appropriate choice the state of the control problem is the triple $Y^H=(X,\Mpro^H,\Qpro^H)$, which takes values in the state space $\mathcal{Y}^H=(0,\infty)\times \R^d\times \R^{d\times d}$.  Note that the conditional variance process $\Qpro^H$ for the information regimes $H=R, Z_n,J$  is deterministic. Thus it can be computed offline and also removed from the state process $Y^H$, as in  \cite{Gabih et al (2022) PowerFixed}.	
	For the regime $H=Z_\lambda$ with random information dates, however, $\Qpro^H$ is a stochastic process and must be included in $Y^H$,  see \cite{Gabih et al (2022) PowerRandom}. For the regime with  full information ($H=F$), the pair $Y^F=(X,\mu)$ is chosen as the state variable, which takes values in the state space $\mathcal{Y}^F=(0,\infty)\times \R^d$. The next lemma shows that the  conditional expectation $\rewardorigin_t^H(\pi)$ defining the performance of an admissible  strategy $\pi$ of the control problem in \eqref{opti_t}  can be expressed as a conditional expectation given the state $Y^H_t$ at time $t$.
	\begin{lemma}		\label{lem:reward_condexp}
		For the conditional expectation $\rewardorigin_t^H(\pi) = \E\left[\utility_{\theta}(X_T^{\pi})~|~ \mathcal{F}^{H,X}_t  \right]$ it holds for $H=R, Z_n,Z_\lambda ,J,F$ and all  $\pi\in\mathcal{A}^{H}$
		\begin{align}
			\label{reward_general} 
			\rewardorigin_t^H(\pi) = \rewardstate^H(t,Y^H_t;\pi)\quad \text{with}\quad  \rewardstate^H(t,y;\pi)=\E\left[\utility_{\theta}(X_T^{\pi})~|~Y_t^H =y\right]. 
		\end{align}	
	\end{lemma}		
	The proof can be found in Appendix \ref{proof_em:reward_condexp}. The function 
	\begin{align}
		\label{reward_state} 
		\rewardstate^H(t,y;\pi)
		= \E\left[\utility_{\theta}(X_T^{\pi})~|~Y^H_t = y  \right],\quad \pi\in\mathcal A^{H},
	\end{align}
	is called \textit{reward function} or \textit{performance criterion} for the strategy $\pi$, and 
	\begin{align}
		\label{value_state} 
		V^H(t,y) = \sup\limits_{\pi\in\mathcal{A}^{H}}   \rewardstate^H(t,y;\pi)
	\end{align}
	\textit{value function} for the family of optimization problems \eqref{opti_t}.

	\section{Well Posedness of the Optimization Problem}
	\label{sec_wellposedness}
	\subsection{Well Posedness}
	\label{Well Posedness}  Solving the  utility maximization problem
	\eqref{opti_org} for the various information regimes
	$H=\HR, Z_n,Z_\lambda ,\HD,\HF$ requires conditions under which the optimization
	problem
	is well posed.  Under these conditions the maximum expected
	utility  of terminal wealth cannot explode in finite time as it is
	the case for  so-called nirvana strategies described in Kim and
	Omberg \cite{Kim and Omberg (1996)} and Angoshtari
	\cite{Angoshtari2013}. Such strategies generate in finite time  a
	terminal wealth with a distribution leading to infinite expected
	utility although the realizations of terminal wealth  may be finite.

	We start by describing the model of the financial market via the
	parameter
	\begin{align}
		\para:=\{T,\theta,d,\volR,\sigma_{\mu},\revspeed,\revlevel,x_0,\driftinitial,\covinitial,
		\filterinitial,\condcovinitial \} \label{modellparameter_1}
	\end{align}
	taking values  in a suitable chosen set of parameter values
	$\modellparameter$. For emphasizing the dependence on the parameter
	$\para$ we rewrite  \eqref{reward_state} and \eqref{value_state}   for $t\in[0,T]$ as
	\begin{align}
		\label{opti_org2} 
		\begin{split}
			\rewardstate^H_\para(t,y;\pi)
			& = \E\left[\utility_{\theta}(X_T^{\pi})~|~Y^H_t = y,\para  \right],\quad \pi\in\mathcal A^{H},\\[0.5ex]
			V^H_\para(t,y) & = \sup\limits_{\pi\in\mathcal{A}^{H}}   \rewardstate^H_\para(t,y;\pi).
		\end{split}
	\end{align}		
	For a given paramter $\para$ we want to study if the
	performance criterion  of the optimization problem \eqref{opti_org2}
	is well-defined in the following sense.
	\begin{definition}
		\label{def_wohlgesteltheit} For a given financial market with
		parameter $\para\in \modellparameter$ we say
		that the utility maximization problem \eqref{opti_org2} for the
		$H$-investor is \textit{well-posed}, if  $V^H_\para(t,y)$ is finite for every $(t,y)\in[0,T]\times \mathcal{Y}^H$.
		The set
		\begin{align}
			\nonumber 
			\modellparameter^{H}=\{\para \in\modellparameter:~
			\text{problem }  \eqref{opti_org2} \text{ is well posed} \}  \subset
			\modellparameter
		\end{align}
		is called  set of \textit{feasible  parameters} of the financial
		market  model for which \eqref{opti_org2} is well-posed.
	\end{definition}
	
	\subsection{Log-Utility and Power  Utility with ${\theta<0}$}
	\label{well_posed_log}
	For power  utility with parameter $\theta<0$ it holds
	$\utility_\theta(x)<0$.  Thus,  $V^H_\para(t,y) $ is bounded above by zero. Further, for the constant buy-and-hold strategy $\pi=0_d$, in which the initial capital is invested exclusively in the bond, it holds $X_t^{\pi}=x_0$ for all $t\in[0,T]$. Then it follows  $\utility_\theta(x_0)=\rewardstate^H_\para(t,y;0_d)\le V^H_\para(t,y)$, and shows, that  $V^H_\para(t,y) $ is also bounded below and is therefore finite.	
	Hence,  the optimization problem is well-posed 	for all model parameters $p\in\modellparameter$ with $\theta<0$. For
	log-utility ($\theta=0$) the utility function is no longer bounded
	from above but  it is shown in 	\cite[Subsec.~4.1]{Gabih et al (2022) PowerFixed}  and \cite[Sec.~4]{Sass et al (2021)} that the value
	function ${V^H_\para(t,y)}$ is bounded from above by some
	positive constant $\rewardconstant^{H} =\rewardconstant^{H}(\para,t,y)$ 
	for  any selection of the 	model parameters in $\para\in\modellparameter$. Hence it
	holds $\{\para\in \modellparameter: \theta\le 0\}\subset
	\modellparameter^{H} $. 
	
	More challenging is the case of power
	utility with positive parameter $\theta\in (0,1)$ which is also not
	bounded from above. That case is investigated in the remainder of
	this section. We note that this approach  can also be applied to
	log-utility leading to an alternative proof of well posedness for
	the maximization of expected log-utility, for details we refer to
	Kondakji \cite[Sec.~4.2]{Kondkaji (2019)}.
	
	\subsection{Power  Utility with ${\theta\in(0,1)}$}
	\label{well_posed_power} For the study of well posedness it will be
	convenient to extend the concept of the fully informed $F$-investor
	who has access to observations of the return and drift process to an
	(artifical) investor who observes also  the sequence of discrete-time expert opinions $(T_k,Z_k)$ and the continuous-time expert opinion process $J$,   as well as the  Wiener processes  $W^R,W^\drift, W^\myzeta$. That
	investor is called $\HG$-investor and defined by the investor
	filtration $\mathbb F^{\HG}=\mathbb{G}$ which is the underlying	filtration to which all stochastic processes of the financial market
	model are adapted. When comparing the $\HF$- and $\HG$-investor, the	additional information from the observation of expert opinions and
	the driving Wiener processes  $W^R,W^\drift, W^\myzeta$ will not
	lead to superior performance of the $\HG$-investor in  the
	considered utility maximization problem,  since the distribution of the wealth process $X^\pi$ is fully determined by the return process $R$ and the drift process $\mu$.  The latter is known to the $G$-investor and does not need to be estimated. 
	Thus, the associated state process $Y^G$ can be chosen as the pair $(X^\pi,\mu)$. However, for the $G$-investor we have  the
	inclusion $\mathbb F^{H}\subset \mathbb F^{\HG}$ for
	$H=\HR, Z_n,Z_\lambda ,\HD,\HF$. Note that for the $F$-investor we only have
	$\mathbb F^{\HR} \subset \mathbb F^{\HF}$, but in general  $\mathbb{F}^{Z_n}, \mathbb{F}^{Z_\lambda},\mathbb F^{\HD} \not\subset \mathbb F^{\HF}$. Analogous to
	the other investors  we define for  $H=\HG$  the set of admissible
	strategies $\mathcal{A}^{\HG}$,  the performance of a strategy
	$\rewardorigin^{\HG}_{t} $, the value $\mathcal{V}^{\HG}_{t} $, the reward function $D_\para^G(t,y,\pi)$ and the value function $V_\para^G(t,y)$
	as in  \eqref{set_admiss_0} and  \eqref{opti_t} through \eqref{value_state}, respectively.
	
	Next we want to derive estimates of the value 
	$\valueorigin_{t} ^{H}$ of the $H$-investor  in terms of the value
	$\valueorigin_{t} ^{\HG}$ of the $\HG$-investor. Let us fix a
	strategy $\pi\in\mathcal A^{H}{ \subset \mathcal A^{G}}$, then  tower property of the
	conditional expectation with  $ \mathcal F_{t}^{H,X}\subset\mathcal	F_{t} ^{\HG,X}$ implies  
	\begin{align}
		\rewardorigin_{t} ^{H}(\pi)&=\E
		\big[\utility_{\theta}(\welth_T^{\pi}) \big{|} \mathcal F_{t}^{H,X}\big]
		=\E \big[ \E  \big[\utility_{\theta}(\welth_T^{\pi}) \big| \mathcal
		F_{t}^{\HG,X} \big] \big{|} \mathcal F_{t} ^{H, X}\big] =\E
		\big[\rewardorigin_{t} ^{\HG}(\pi)\big{|}\mathcal F_{t} ^{H,X}\big].
	\end{align}
	Using Lemma \ref{lem:reward_condexp} yields for all information regimes $H$  
	\begin{align}
		\rewardorigin_{t} ^{H}(\pi)&= \rewardstate_\para^H(t,Y_t^H,\pi)= \E\big[\rewardstate_\para^G(t,Y_t^G,\pi) \big{|}\mathcal F_{t} ^{H,X}\big]
		= \E\big[\rewardstate_\para^G(t,Y_t^G,\pi) \big{|}Y_t^H\big].
	\end{align}
	Taking  supremum over all admissible strategies in $\mathcal	A^{H}$  it follows for all fixed  $Y^H_t=y\in\mathcal{Y}^H$ 
	\begin{align}
		V_\para^H(t,y)& = \sup_{\pi\in\mathcal A^{H}} \rewardstate_\para^H(t,y,\pi) = 
		\sup_{\pi\in\mathcal A^{H}} \E\big[\rewardstate_\para^G(t,Y_t^G,\pi) \big{|}Y_t^H=y\big].
	\end{align}
	Using  $\mathcal A^{H}\subset \mathcal A^{G}$ and properties of the supremum we find the estimates 
	\begin{align}
		V_\para^H(t,y)&  \le \E \Big[\sup_{\pi\in\mathcal A^{H}}	\rewardstate_\para^G(t,Y_t^G,\pi) \big{|}Y_t^H=y\Big] \\
		&\leq  \E \Big[\sup_{\pi\in\mathcal A^{G}}	\rewardstate_\para^G(t,Y_t^G,\pi) \big{|}Y_t^H=y\Big]
		= \E \big[V_\para^G(t,Y_t^G)\big{|}Y_t^H=y\big].
		\label{bounded_2}
	\end{align}

	In the sequel we will derive conditions under which
	$V_\para^G(t,y)$ with $y=(x,m)$  is bounded for any fixed $t\in[0,T]$ and $X_t^\pi=x,\mu_t=m$.  Then estimate
	\eqref{bounded_2} will allow us to derive conditions for the
	boundedness of $V_\para^H(t,y)$ for the other information
	regimes $H$. We will need the following lemma where we denote by
	$\drift_u^{t,\filter}$  the drift at time $u\in[t,T]$ starting at
	time $t\in[0,T]$ from $\filter\in\R^d$ . 	The  proof is given in Appendix
	\ref{proof_lemma_Helplemma_Psi_Expect}.
	\begin{lemma}
		\label{Helplemma_Psi_Expect} Let $\cpsi\in\mathbb R\backslash\{0\}$, $t\in[0,T)], z>0,$ and the stochastic process
		$(\Psi_s^{t,z,\filter})_{s\in[t,T]}$ be defined by
		\begin{align}
			\Psi_s^{t,z,\filter}:=z\exp\Big\{ \cpsi \int\nolimits_t^s
			(\drift_u^{t,\filter})^{\top}
			\Sigma_{\HR}^{-1}\drift_u^{t,\filter}\; du \Big\} \label{Psi_pro}
		\end{align}
		and the function  $d:[0,T]\times \R^d\to\R_{+}$ be  defined by
		$~d(t,\filter):=\E \big[\Psi_T^{t,1,\filter}\big],~$ for
		$t\in[0,T]$ and $\filter\in\R^d$. Then it holds
		\begin{align}
			d(t,\filter):=\exp\Big\{\filter^{\top} \Abound(t) \filter+\Bbound^{\top} (t)\filter +\Cbound (t) \Big\}.
			\label{Psi_4}
		\end{align}
		Here  $\Abound(t)$, $\Bbound(t)$ and $\Cbound(t)$ are functions in
		$t\in[0,T]$ taking values in $\mathbb R^{{\nAktien}\times
			{\nAktien}}$, $\mathbb R^{\nAktien}$ and $\mathbb
		R$, respectively, satisfying the following system of ODEs
		\begin{align}
			\label{A_bound}
			\frac{d \Abound (t)}{dt}&=-2 \Abound (t)\Sigma_{\drift}\Abound (t)+\revspeed^{\top} \Abound (t)+\Abound (t)\revspeed-\cpsi\Sigma_R^{-1}, & \Abound(T)&=0_{\nAktien\times\nAktien},\\[2ex]
			\label{B_bound}
			\frac{d\Bbound(t)}{dt}&=-2\Abound(t)\revspeed\overline{\drift}+\left[\revspeed^{\top}-2\Abound(t)\Sigma_{\drift}\right] \Bbound(t), &\Bbound(T)&=0_{\nAktien}, \\[2ex]
			\label{C_bound}
			\frac{d\Cbound(t)}{dt}&=-\frac{1}{2}\Bbound^{\top}(t)
			\Sigma_{\drift}
			\Bbound(t)-\Bbound^{\top}(t)\revspeed\overline{\drift}-tr\{
			\Sigma_{\drift} \Abound(t)\},&\Cbound(T) &=0.
		\end{align}
	\end{lemma}
	Note that equation \eqref{A_bound} is a Riccati equation  for the 	symmetric matrix-valued function  $\Abound$,  while equation
	\eqref{B_bound} is  a system of $d$ linear differential equations 	whose solution $\Bbound$ is obtained given $\Abound$. Finally, given
	$\Abound$ and $\Bbound$ the  scalar function $\Cbound$ is 	obtained by integrating the right hand side of \eqref{C_bound}.

	\begin{theorem}
		\label{theo_bound_V} For a model parameter $\para$ with
		$\theta\in(0,1)$ the value function of the $G$-investor  satisfies for $y=(x,m)\in \mathcal{Y}^G=(0,\infty)\times \R^d$
		\begin{align}
			\label{VG_esti}
			V_\para ^G(t,y)\leq \frac{x^{\theta}}{\theta} \,
			d^{1-\theta}( t,m),
		\end{align}
		where the function $d:[0,T]\times\R^d\to \R_{+}$ is given by
		\eqref{Psi_4} for $\cpsi=\frac{\theta}{2(1-\theta)^2}$.
	\end{theorem}
	\begin{proof}
		The  proof is given in Appendix \ref{proof_lemma_lemma_bound_V}.
	\end{proof}
	The last theorem  together with the fact that for $\theta\le 0$ the	problem is well-posed (see the reasoning at the beginning of this
	section) allows to give the following characterization of the	$\HG$-investor's set $\modellparameter^{\HG}$ of feasible model
	parameters by the inclusion
	$\underline{\modellparameter}^{\HG}\subset \modellparameter^{\HG}$
	where  
	\begin{align}
		\begin{aligned}
			\underline{\modellparameter}^{\HG} & =\{\para\in\modellparameter:
			\theta\in(0,1) \text{ and }   d(t,m) \text{ given in
				\eqref{Psi_4}  is finite  for every fixed }  \\
			& ~~~~~~ (t,m)\in[0,T]\times \R^d\} \cup \{\para\in\modellparameter:
			\theta\le 0\}.
		\end{aligned}
		\label{modellparameter_3}
	\end{align}

	We are now in a position to characterize the set of feasible model
	parameters $\modellparameter^H$  for $H=\HR, Z_n,Z_\lambda ,\HD,\HF$  by combining  the estimate \eqref{VG_esti} for $V^{\HG}_\para(t,y)$  from
	Theorem \ref{theo_bound_V} with  \eqref{bounded_2} stating that    $V_\para^H(t,y) \leq  
	\E \big[V_\para^G(t,Y_t^G)\big{|}Y_t^H=y\big]$. Recall, that for the partially informed investors ($H=R, Z_n,Z_\lambda,J$) the state process  is $Y_t^H=(X^\pi,\Mpro^H,\Qpro^H)$. For the  $G$- and $F$-investor it is $Y_t^{G/F}=(X^\pi,\mu)$. Thus, substituting  the estimate \eqref{VG_esti} for $\valuefkt_\para^\HG$  into \eqref{bounded_2}   yields for	the partially informed investors for $y=(x,m,q)$
	\begin{align}		
		V_\para^{H}(t,y)\leq
		\frac{x^{\theta}}{\theta}\E \Big[d^{{1-\theta}}(t,\drift_t) \big| \Mpro^H_t=m,\Qpro_t^H=q\Big],\quad H=R,Z,J,
		\label{bound_value_H}
	\end{align}
	and for the fully informed investor for $y=(x,m)$
	\begin{align}		
		V_\para^{F}(t,y)\leq
		\frac{x^{\theta}}{\theta}\E \Big[d^{{1-\theta}}(t,\drift_t) \big| \drift_t=m\Big].
		\label{bound_value_F}
	\end{align}
	
	\medskip
	\paragraph{Well posedness for full information ($\mathbf{H=\HF}$)}
	For the $\HF$-investor the drift is known and from  
	inequality \eqref{bound_value_F} it follows  for $y=(x,m)\in \mathcal{Y}^F=(0,\infty)\times \R^d$
	\begin{align}
		V^{\HF}_\para(t,y)\leq  	\frac{x^{\theta}}{\theta} d^{1-\theta}( t,m), \nonumber
	\end{align}
	which implies that the inclusion given in \eqref{modellparameter_3}
	for $\modellparameter^{\HG}$ also holds for the  set of feasible
	model  parameters  $\modellparameter^{\HF}$  for the $\HF$-investor.

	The  restrictions in  \eqref{modellparameter_3} to the  feasible model parameters $\para$  for $\theta\in(0,1)$  are given implicitly via the boundedness of $d(t,m)$  where $d$ is given in \eqref{Psi_4}. They can be further analysed by
	studying conditions for non-explosive solutions of Riccati equation
	\eqref{A_bound} for the matrix-valued function $\Abound$  on the
	investment horizon $[0,T]$.  We refer to the specialist literature  such as  Kucera \cite{Kucera (1973)}, Roduner \cite{Roduner (1994)}, and Wonham \cite{Wonham (1968)}, and  the references  therein. The boundedness of
	the solution to  \eqref{A_bound} carries over to the boundedness of
	the solution to the linear  differential equation \eqref{B_bound} for
	$\Bbound$ and also to $\Cbound$ which  is obtained by integrating
	the right hand side of \eqref{C_bound}.   Thus we obtain
	\begin{cor}[Sufficient condition for well posedness, full information]
		\label{suff_cond_full}	
		\  \\
		The utility maximization problem \eqref{opti_org2} for the fully informed
		$F$-investor is well-posed for all parameters $\para\in	
		\underline{\modellparameter}^{\HF}\subset \modellparameter^{\HF}$
		where 
		\begin{align}
			\underline{\modellparameter}^{F}= & \Big\{\para\in\modellparameter:
			\theta\in(0,1),~ \Abound \text{ is bounded on } [0,T] ~\text{ for } \cpsi=\frac{\theta}{2(1-\theta)^2}  \Big\}\\
			&  \cup \{\para\in\modellparameter:
			\theta\le 0\}, \label{modellparameter_full}
		\end{align}
		and  $\Abound$ is the solution to  Riccati equation \eqref{A_bound}.	
	\end{cor}

	We observe that  the sufficient condition derived in Corollary \ref{suff_cond_full}	does not depend on the initial values of the state process $(x_0,m_0)$ but only on $T$, and the constant parameters $\theta,\volR,\sigma_{\mu},\revspeed$. Note that, if the  solution $A_\cpsi$ to the terminal value problem for the Riccati equation \eqref{A_bound}  does not explode on  $[0,T]$, i.e.,  it is bounded, then it also does not explode on $[t,T]$ for any $t\in[0,T]$. Thus, for the well posedness of the $F$-investor's optimization problem \eqref{opti_org2} it is sufficient that $A_\cpsi(t)$ is bounded for $t=0$. Below we will see that for utility maximization problems under partial information we need stronger conditions.
	
	It is well known, that in general
	closed-form solutions of Riccati differential equations are
	available only for the one-dimensional case ($d=1$). More details
	about this special case can be found below in
	Subsec.~\ref{WellPosedSpecialCase}.

	\medskip
	\paragraph{Well posedness for partial information ($\mathbf{H=\HR, Z_n,Z_\lambda ,\HD}$)}
	For the partially informed investors  the conditional expectation of  the random variable $d(t,\drift_t)$ given the filter estimates of the hidden drift $\mu_t$  in \eqref{bound_value_H}  has to be analyzed, taking advantage of the fact that the  conditional distribution of $\mu_t$  is Gaussian and completely characterized by the
	conditional mean $\Mpro_t^{H}$ and the conditional covariance	$\Qpro_t^{H}$.   
	The result is given below	in  Theorem \ref{theorem_partial_Inv} for which we need the
	following lemma  which is proven   in Appendix	\ref{proof_moment_quadratic_form}.

	\begin{lemma}
		\label{moment_quadratic_form}
		Let $U,\Sigma$ be symmetric $d\times d$ matrices, $\Sigma$ positive semidefinite with an associated decomposition $\Sigma=\Sigmafac \Sigmafac^\top$ with a $d\times d$-matrix $P$.  Further, let the eigenvalues of $\Sigma{U}$ be denoted by  $\lambda_1,\ldots,\lambda_d$.
		
		\begin{enumerate}
			\item  The eigenvalues $\lambda_1,\ldots,\lambda_d$ of $\Sigma{U}$  are also the eigenvalues of  $\Sigmafac^{\top}U\Sigmafac$. They  are all real.
			\item
			Let   $K=I_{\nAktien}-2\Sigma{U}$,  $a\in \R^d$,   $Y\sim\mathcal N(0_d,\Sigma)$ be a  $d$-dimensional zero-mean Gaussian random vector with covariance matrix $\Sigma$,   and $\quadform$ the real-valued random variable defined by the quadratic form $\quadform=Y^{\top} U Y+a^{\top} Y$.
			
			If   $\lambda_i<\frac{1}{2}$~ for all $i=1,\ldots,d$, then it holds for the exponential moment of $\quadform$
			\begin{align}
				\label{quadratic form}
				\E[\,e ^\quadform\,]=\E \big[ \exp{(Y^{\top} U Y+a^{\top} Y)} \big]=(\det(K))^{-1/2}\exp\Big\{\frac{1}{2}a^{\top}K^{-1}{\Sigma}a\Big\}.   
			\end{align}
			\item For the terms of the right hand side of \eqref{quadratic form} it holds
			\begin{align}
				\label{quadratic form_eigenvalue}
				\det(K)= \prod\limits_{j=1}^d(1-2\lambda_j) \quad \text{and}\quad a^{\top}K^{-1}{\Sigma} a= \sum\limits_{j=1}^d c_j^2(1-2\lambda_j)^{-1},
			\end{align}
			where $c_1,\ldots,c_d$ are the entries of the vector  $c=D^{\top}\Sigmafac^{\top}a$, and  $D$ is the orthogonal $d\times d$-matrix  diagonalizing the symmetric  matrix $\Sigmafac^{\top}U\Sigmafac$ such that it holds $\Sigmafac^{\top}U\Sigmafac=D\Lambda D^\top$ with $\Lambda=\diag(\lambda_1,\ldots,\lambda_d)$. 
		\end{enumerate}		
	\end{lemma}
	
	\begin{remark}
		The expressions in 	\eqref{quadratic form_eigenvalue} are helpful for the actual computation of the expectation in \eqref{quadratic form}. For large dimensions $d$ and a covariance matrix $\Sigma$ of  low-rank $r\ll d$  the computational efficiency can be improved by  working with a low-rank decomposition $\Sigma=\Sigmafac_r \Sigmafac_r^\top$ with a $d\times r$-matrix $\Sigmafac_r$ and an eigenvalue decomposition of the $r\times r$-matrix $\Sigmafac_r^{\top}U\Sigmafac_r=D_r\Lambda_r D_r^\top$ with an orthogonal $r\times r$-matrix $D_r$ and the diagonal matrix $\Lambda_r$ obtained form $\Lambda$ by removing  $d-r$ zero eigenvalues on the diagonal.
		
	\end{remark}

	\begin{theorem}
		\label{theorem_partial_Inv}
		Let for the information regimes $H=\HR, Z_n,Z_\lambda ,\HD$ with partial information and $t\in[0,T]$ denote  the conditional mean and covariance describing the Gaussian conditional distribution  of the  the drift $\mu_{t}$ given $\mathcal{F}_{t}^H$ by $\Mpro^H_t=m\in\R^d$ and $\Qpro^H_t=q\in\mathbb R^{d\times d}$. 
		Further, let the solution  $\Abound$  to the Riccati equation \eqref{A_bound} for $\cpsi=\frac{\theta}{2(1-\theta)^2}$ be bounded on $[0,T]$,
		and assume that all eigenvalues of ${K= K(t)}=I_{\nAktien}-2{ (1-\theta)}{q\Abound(t)}$  are positive. \\Then it holds  for all $x>0$ and $y=(x,m,q)$
		\begin{align}
			V^{H}_\para(t,y)&\leq  \frac{x^{\theta}}{\theta} d^{{1-\theta}}(t,m) (\det({K}))^{-1/2} {\exp\Big\{\frac{1}{2}}{a}^{\top} {K}^{-1}qa\Big\},
			\label{bound_7}
		\end{align}
		where
		$a={a(t)= (1-\theta)}2\Abound({ t}){m} +\Bbound(t)$   and $d(t,m)$ is given in	\eqref{Psi_4}.
	\end{theorem}
	The proof can be found  in Appendix	\ref{proof_theorem_partial_Inv}.
	We observe that the upper bound for $V^{H}_\para(t,y)$ given in Theorem \ref{theorem_partial_Inv} is finite if $d(t,m)$ is finite and  all eigenvalues of   ${K}=I_{\nAktien}-2{ (1-\theta)}q\Abound(t)$ are positive.  As in the full information case, the first condition is satisfied if the  Riccati equation \eqref{A_bound}  for  $\Abound$  are bounded on  $[0,T]$. Recall, this depends only on the choice of the constant model parameters $\theta,\volR,\sigma_{\mu},\revspeed$ but not on $\Mpro^H_{t}=m$, and  it implies boundedness of $\Abound$  on $[t,T]$ for all $t\in[0,T]$. 
	However, the second condition is for $\theta\in(0,1)$ an additional restriction for the partially informed case and states that the conditional covariance $\Qpro^H_{t}=q$ is not ``too large'' such that  all eigenvalues of ${K}$ are positive. We have to require that the conditional covariance $\Qpro^H$ process starting at $t=0$ with $\Qpro^H_0=\condcovinitial$ is such that 
	the eigenvalues of ${K}= K(t)=I_{\nAktien}-2{ (1-\theta)}\Qpro^H_t\Abound(t)$ remain positive on the entire time interval $[0,T]$. 		
	Note that for the regimes $H=R,J$ and for $H=Z_n$ with discrete-time expert opinions at fixed arrival times, $\Qpro^Z$ a deterministic function  and  fully specified by its initial value $\condcovinitial$. However, for the regime  $H=Z_\lambda$ with discrete-time expert opinions at random arrival times, $\Qpro^Z$ is a stochastic process, it depends on the random timing of expert's views and is not only specified by its initial value $\condcovinitial$.
	
	For $\theta\in(0,1)$ it is known that if the solution $\Abound$ exists on [0,T], it is  symmetric and positive semidefinite,
	see  Roduner \cite[Theorem 1.2]{Roduner (1994)}. Further, the conditional covariance $\Qpro^H$ is also symmetric and positive semidefinite.  However, the product $\Qpro^H\Abound$ of the two symmetric matrices is generally no longer symmetric, and the properties of such matrices may no longer apply. But it is known that $\Qpro^H\Abound$ has the same eigenvalues as $D=S\Qpro^HS$ with $S$ denoting the unique symmetric and positive semidefinite square root of $\Abound $, that is $\Abound=SS $. Since $D$ is symmetric and  positive semidefinite its eigenvalues and therefore the eigenvalues of $\Qpro^H\Abound$ are nonnegative. Finally, let $\lambda=\lambda( \Qpro^H\Abound)\ge 0$ be an arbitrary eigenvalue of $ \Qpro^H\Abound$. Then  $1-2{ (1-\theta)}\lambda$  is an eigenvalue of  $K=I_{d}-2{ (1-\theta)}\Qpro^H\Abound$.
	Thus, the  condition that the eigenvalues of $K$ are positive implies that  all  eigenvalues  of   $\Qpro^H_t\Abound(t)$  are required to be  strictly smaller than $\frac{1}{2{ (1-\theta)}}$ for all $t\in[0,T]$. 
	Let  $\lambda_{\max}(G)$ denote the largest eigenvalue of a generic matrix $G$ with real and nonnegative eigenvalues, then this condition can be stated as 
	\begin{align}
		\label{suff_cond_partial_1}
		\lambda_{\max}(\Qpro^H_t\Abound(t))<\frac{1}{2{ (1-\theta)}}, \quad \text{ for all } t\in[0,T].
	\end{align}
	Summarizing, from the above theorem we deduce the following sufficient condition for well posedness.
	\begin{cor}[Sufficient condition for well posedness, partial information]
		\label{suff_cond_partial}	\ \\
		The utility maximization problem \eqref{opti_org2} for partially informed 
		$H$-investors, $H=\HR, Z_n,Z_\lambda ,\HD$, is well-posed for all parameters $\para\in	
		\underline{\modellparameter}^{H}\subset \modellparameter^{H}$
		where 
		\begin{align}
			\nonumber
			\underline{\modellparameter}^{H}=  \Big\{&\para\in\modellparameter:~
			\theta\in(0,1),~ \Abound~ \text{is  bounded on } [0,T] ~\text{ for } \cpsi=\frac{\theta}{2(1-\theta)^2},
			\\&  			
			\text{and }~  \lambda_{\max}(\Qpro^H_t\Abound(t))<\frac{1}{2{ (1-\theta)}}, ~ \text{ for all } t\in[0,T]				 
			\Big\}   
			\cup \big\{\para\in\modellparameter:	\theta\le 0\big\}
			\label{modellparameter_partial}.
		\end{align}		
	\end{cor}

	\subsection{Market Models With a Single Risky Asset}
	\label{WellPosedSpecialCase} The above conditions for the well
	posedness  given in terms of the boundedness of  $\Abound(t)$ on $[0,T]$, and condition \eqref{suff_cond_partial_1} to the eigenvalues of $\Qpro^H_t\Abound(t)$ are quite abstract and its verification requires that the solution of the Riccati ODE \eqref{A_bound} is bounded on $[0,T]$. While in the multi-dimensional case Riccati ODEs in general can be solved only numerically these equations  enjoy a closed-form solution in the one-dimensional case. 
	This allows to give more explicit characterizations of the set of feasible parameters for market models with a single risky asset only.
	The following lemma gives explicit conditions to the model
	parameters under which  \eqref{A_bound} has a  bounded solution on
	$[0,T]$. For the proof we refer to Kondakji \cite[Lemma A.1.3, A.2.2
	and  A.2.3]{Kondkaji (2019)}

	\begin{lemma}
		\label{lemma_Beding_beschr}
		Let $\nAktien=1$,  $\theta\in(0,1), \cpsi=\frac{\theta}{2(1-\theta)^2}$,  and
		\begin{align}
			\label{Diskriminante_Barier}
			\Delta_{\cpsi}=4\revspeed^2\Big(1-2 \cpsi\Big(\frac{\sigma_{\drift}}{\revspeed\sigma_{\HR}}\Big)^2 \Big)
			\quad \text{and} \quad
			\delta_{\cpsi}:=\frac{1}{2}\sqrt{|\Delta_{\cpsi}|}.
		\end{align}
		Then it holds for the Riccati differential equation \eqref{A_bound} on $[0,T]$
		\begin{enumerate}
			\item For $\Delta_{\cpsi}\geq0$  there is a bounded solution for all $T>0$.
			\item For $\Delta_{\cpsi}<0$ a bounded solution exists only if $T<T_{\cpsi}^{E}$ with the explosion time
			\begin{align}
				T^E_\gamma :=\frac{1}{\delta_{\cpsi}}
				\Big(\frac{\pi}{2}+\arctan\frac{\revspeed}{\delta_{\cpsi}}\Big).
				\label{t_Explosion_AV}
			\end{align}
		\end{enumerate}
	\end{lemma}

	The above lemma allows to give more explicit sufficient conditions for well posedness given for the general multi-dimensional case  in Corollary \ref{suff_cond_full}	and \ref{suff_cond_partial}. They can be formulated in terms of the parameters $\revspeed,\sigma_\drift,\sigma_{\HR}$ describing the variance of the asset price and drift process, the investment horizon $T$, the parameter $\theta$ of the utility function and the  conditional covariance process $\Qpro^H$. Analyzing the inequality $T<T^E$ and \eqref{t_Explosion_AV} we obtain
	\begin{cor}[Sufficient condition for well posedness, single risky asset]
		\label{suff_cond_d1}
		\  \\
		Let $d=1$, $\cpsi=\frac{\theta}{2(1-\theta)^2}$ and $\Delta_{\cpsi}, T_{\cpsi}^{E}$ as given in \eqref{Diskriminante_Barier} and \eqref{t_Explosion_AV}, respectively.
		\begin{enumerate}
			\item 
			The utility maximization problem \eqref{opti_org2} for the \textbf{fully informed}
			$F$-investor is well-posed for all parameters $\para\in	
			\underline{\modellparameter}^{\HF}\subset \modellparameter^{\HF}$
			where 
			\begin{align}
				\label{modellparameter_full_d1}		
				\underline{\modellparameter}^{\HF}= & 			 \big\{\para\in\modellparameter:
				\theta\in(0,1), \revspeed,	\sigma_{\drift}, \sigma_{\HR}~ 
				\text{such that  either }~\Delta_{\cpsi}\geq0, ~ 
				\text{ or } ~\Delta_{\cpsi}<0 \text{ and }  T<T_{\cpsi}^{E}\big\}\\[-2ex]
				& \cup 	 \big\{\para\in\modellparameter:	\theta\le 0\big\} 	.
			\end{align}
			\item
			The utility maximization problem \eqref{opti_org2} for the \textbf{partially  informed}
			$H$-investor with  $H=\HR, Z_n,Z_\lambda ,\HD$ is well-posed for all parameters $\para\in	
			\underline{\modellparameter}^{H}\subset \modellparameter^{H}$
			where 
			\begin{align}
				\label{modellparameter_partial_d1}						
				\begin{aligned}						
					\underline{\modellparameter}^{H}=  \big\{\para\in\modellparameter: &
					\theta\in(0,1), \revspeed,	\sigma_{\drift}, \sigma_{\HR} \text{ such that  either }\Delta_{\cpsi}\geq0, ~ 
					\text{ or } ~\Delta_{\cpsi}<0 \text{ and }  T<T_{\cpsi}^{E},\\					
					&	   \Qpro^H(t) <\frac{1}{2{ (1-\theta)}\Abound(t)}\text{ on } [0,T],  \big\} \cup
					\big\{\para\in\modellparameter:	\theta\le 0\big\} . 
				\end{aligned}					
			\end{align}
		\end{enumerate}
	\end{cor}

	\section{Numerical Results}
	\label{numeric_result} In this section we illustrate the theoretical
	findings of the previous sections by results of some numerical
	experiments.
	They are based on a stock market model where the unobservable drift
	$(\drift_t)_{t\in[0,T]}$ follows an Ornstein-Uhlenbeck process as
	given in \eqref{drift} whereas the volatility is known and constant.
	For simplicity, we assume that there is only one risky asset in the
	market, i.e. $\nAktien=1$.
	If not stated otherwise, our numerical experiments are based on model
	parameters as given in Table \ref{parameter}.
	\begin{table}[ht]
		\begin{tabular}{|lll|r||ll|r|}
			\hline
			\rule{0mm}{2.5ex}%
			Drift & mean-reversion level& $\revlevel$ & $0$      &    Investment horizon & $T$ & $1~$ year
			\\ \hline \rule{0mm}{2.5ex}%
			& mean-reversion speed& $\revspeed$ & $3$ &    Power utility parameter  & $\theta$ & $0.3$
			\\ \hline
			& volatility  &$\sigma_{\mu}$  &$1$ &      Volatility of stock & $\volR$ & $0.25$
			\\\hline \rule{0mm}{2.5ex}%
			& mean of $\drift_0$  &   $\driftinitial$  & $\revlevel=0$&    Initial estimate  & $\filterinitial=\driftinitial$ & $0$
			\\ \hline
			&variance of $\drift_0$  &   $\covinitial$  & $\frac{\sigma_{\drift}^2}{2\revspeed}=0.1\overline{6}$ &  &   $\condcovinitial=\covinitial$ & $0.1\overline{6}$
			\\ \hline
		\end{tabular}
		\\[1ex]
		\centering \caption{\label{parameter}
			\small Model parameters for the numerical experiments
		}
	\end{table}
	
	\begin{figure}[t!h]
		\hspace*{-0mm}
		\includegraphics[width=1\textwidth]{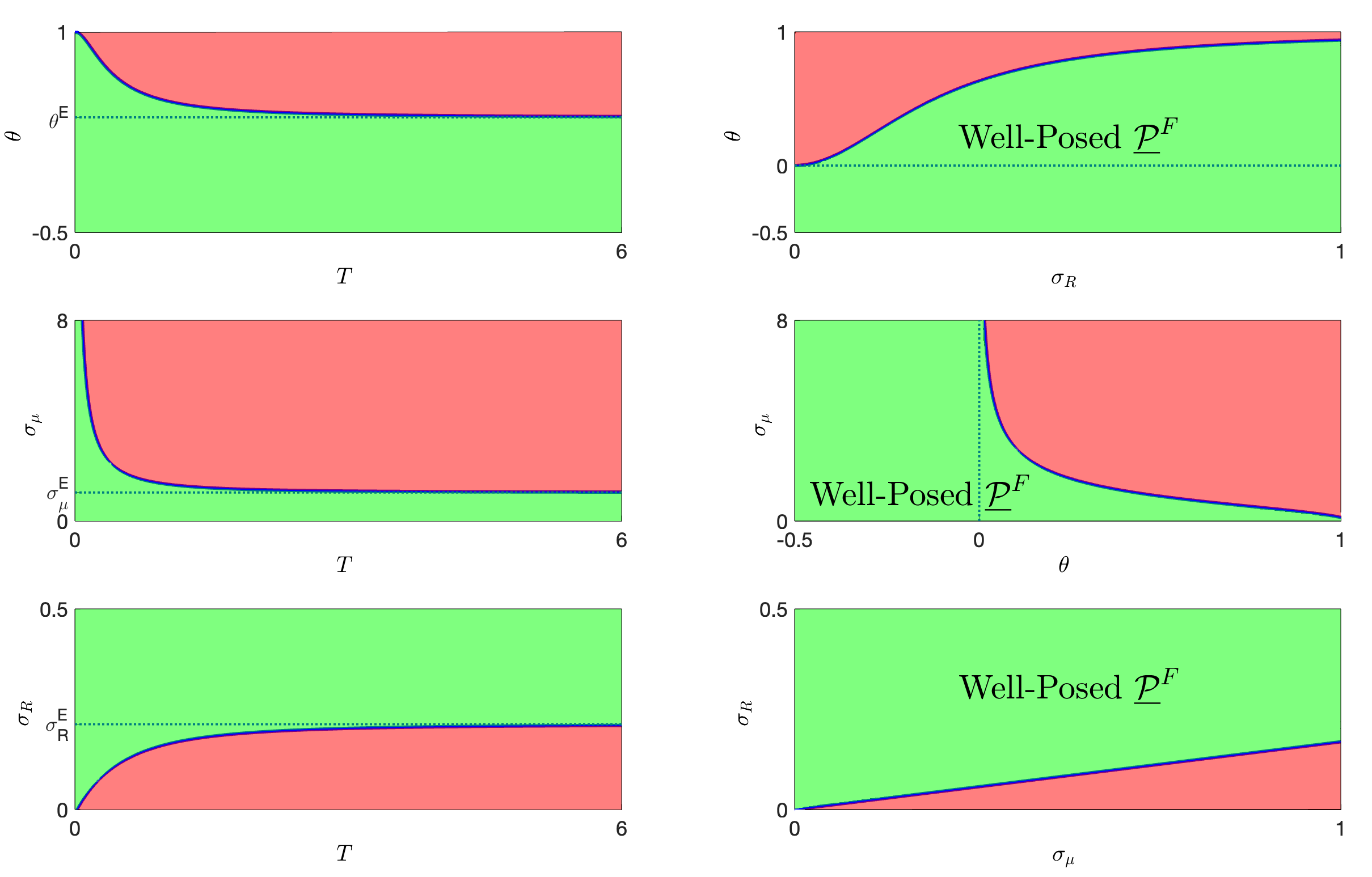}
		\centering 
		\caption{\label{nirvana_abbild}
			\parbox[t]{0.7\textwidth}{
				\small
				Subset $\underline{\modellparameter}^{\HF}$ of the set of feasible parameters $\modellparameter^{\HF}$   
				\newline
				depending on~~           $\theta~~$ and $T$  (top left),    \hspace*{2em}$\theta$ and $\volR$  (top right),  \newline
				\phantom{depending on~~}           $\sigma_{\mu}$ and $T$  (middle left),     $\sigma_{\mu}$ and $\theta$ ~ (middle right), \newline
				\phantom{depending on~~}           $\volR$ and $T$  (bottom left),    $\volR$ and $\sigma_{\mu}$  (bottom right). \newline
				The other parameters are given in Table \ref{parameter}.        
		}}
	\end{figure}
	
	In Section \ref{sec_wellposedness} we have specified sufficient conditions to the model parameters for which the optimization problem is
	well-posed. For market models with a single risky asset  these conditions are given Corollary \ref{suff_cond_d1}. In Figure \ref{nirvana_abbild} we visualize the subset $\underline{\modellparameter}^{\HF}$ of the set of feasible parameters $\underline{\modellparameter}^{\HF}$ for which well posedness for the utility maximization problem of the fully informed investor can be guaranteed.  In particular, we show the dependence of $\underline{\modellparameter}^{\HF}$  on the
	investment horizon $T$, the power utility parameter $\theta$, the
	volatility $\sigma_{\drift}$ of the drift
	and the volatility $\sigma_{\HR}$ of the stock price.
	
	The two top panels show the  subset $\underline{\modellparameter}^{\HF}$ depending
	on $\theta, T$ and $\sigma_R$. It can be seen that for negative
	$\theta$, i.e.~for investors which are more risk averse than the
	log-utility investor,  the optimization problem is always
	well-posed. Moreover, the top left panel shows that for the selected
	parameters the optimization problem is well-posed for all $T>0$ if
	the parameter $\theta$ does not exceed the critical value
	$\theta^{E}\approx 0.36$. For  $\theta>
	\theta^{E}$, i.e.~for investors with sufficiently small risk-aversion, the optimization problem is no longer
	well-posed for all investment horizons $T$, but only up to a
	critical investment horizon $T^E=T^E(\theta)$ depending on $\theta$
	and given in \eqref{t_Explosion_AV}. The larger $\theta$, the
	smaller is that critical  investment horizon $T(\theta)$. For the
	limiting case $\theta\rightarrow 1$ it holds $T^E(\theta)\rightarrow
	0$. The top right panel shows for an investment horizon fixed to
	$T=1$ the  subset $\underline{\modellparameter}^{\HF}$  depending on  $\theta$
	and the volatility $\volR$ of stock price. It can be seen that  larger
	values for the stock volatility allow to choose larger values of
	$\theta$.

	The two panels in the middle illustrate the influence of the drift
	volatility $\sigma_{\mu}$ on  the  subset $\underline{\modellparameter}^{\HF}$ . The
	left panel shows  that the optimization problem is well-posed for
	all $T>0$ as long as the volatility $\sigma_{\drift}$ of the drift
	does not exceed the critical value $\sigma_{\drift}^{E}\approx
	1.15$. For volatilities $\sigma_{\drift}>\sigma_{\drift}^{E}$ the
	optimization problem is well-posed only for investment horizons $T$
	smaller than the critical horizon $T^E=T^E(\sigma_{\drift})$ that
	depends on $\sigma_{\drift}$ and is given in \eqref{t_Explosion_AV}.
	
	In the right panel  we investigate for fixed investment horizon
	$T=1$ the dependence of $\underline{\modellparameter}^{\HF}$ on  the drift
	volatility $\sigma_{\mu}$   and the power utility parameter
	$\theta$. While for $\theta<0$ there are no further restrictions on
	the parameters, this is no longer true for $\theta\in(0,1)$. The
	larger  the volatility $\sigma_{\mu}$ the smaller one has to choose
	$\theta$.
	
	The bottom two panels illustrate the influence of stock  volatility
	$\volR$  on the subset $\underline{\modellparameter}^{\HF}$ . In contrast to the
	volatility  $\sigma_{\drift}$ of the drift, smaller values of
	$\sigma_{\HR}$  imply  that the optimization problem is well-posed
	only for smaller $T$  as it can be seen in the bottom left panel. If
	$\sigma_{\HR}$ does not exceed the critical value
	$\sigma_{\HR}^E\approx 0.22$, then the optimization problem is
	well-posed only up to a critical  investment horizon
	$T^E=T^E(\sigma_{\HR})$ that depends on $\sigma_{\HR}$. The larger
	$\volR$, the larger the horizon can be set. However, for  $\volR$
	exceeding the critical value $\sigma_{\HR}^E$ the control problem is
	well defined for any horizon time  $T>0$. The bottom right panel
	shows the dependence of  $\underline{\modellparameter}^{\HF}$ on the  two
	volatilities $\volR$ and $\sigma_{\mu}$. Note that the two regions
	are separated by a straight line as it can be deduced from \eqref{Diskriminante_Barier} and 
	\eqref{t_Explosion_AV}. 
	
	Finally, we consider the case of a partial informed investor. Then  an additional condition  on the  covariance process
	$\Qpro^H$ of the filter has to imposed to ensure well posedness. We refer to Corollary \ref{suff_cond_partial} for the multi-dimensional case  and to Corollary \ref{suff_cond_d1} and \eqref{modellparameter_partial_d1} for the special case of markets with a single risky asset considered here.  The sufficient condition requires  that
	$\Qpro^H_t<1/(2{ (1-\theta)}\Abound(t))$ which is satisfied for the  model parameters from Table
	\ref{parameter}.  First, it is known that  $\Abound$ is decreasing on $[0,T]$. Second, properties of the conditional covariance process (see \cite{Gabih et al (2022) PowerFixed,Sass et al (2017),Sass et al (2021)}) imply  $\Qpro^H_t\le \Qpro^R_t$. Further, for $\condcovinitial>\Qpro_\infty^R=\lim_{t\to\infty} \Qpro_t^R$ we have that $\Qpro^R$ is decreasing.  According to  \cite[Prop.~4.6]{Gabih et al (2014)} it holds $\Qpro_\infty^R=\sigma_R\sqrt{\sigma_R^2\kappa^2 +\sigma_\mu^2} -\kappa \sigma_R^2=0.125 $ and thus   $\Qpro_\infty^R<\condcovinitial=0.1\overline{6}$ which yields that  we have  $\Qpro^H_t\le \condcovinitial\le \frac{1}{2{ (1-\theta)}\Abound(t)}$ since $\frac{1}{2{ (1-\theta)}\Abound(t)}\ge \frac{1}{2{ (1-\theta)}\Abound(0)}={ 0.964}$. This shows that the problem is well-posed.
	
	\section*{Conclusion}
	The paper derived sufficient conditions for the  well posedness of power utility maximization problems under full and partial information on the not directly observable drift of risky assets.  For power utility with relative risk aversion smaller than that of log-utility these conditions ensure the absence of nirvana strategies as well as  bounded value functions arising in the solution with dynamic programming techniques. They  lead to restrictions on the choice of the model parameters  such as the investment horizon and  the risk aversion parameter of the power utility function, parameters controlling the variance of the asset price and drift processes.
	For the fully informed investor well posedness does not depend on the choice of the parameters $x_0,\driftinitial,\covinitial$ defining the initial values of the state process $Y^F=(X^\pi,\drift)$. However, and somewhat surprisingly,  for the partially informed investors ($H=R, Z_n,Z_\lambda,J$), well posedness is only guaranteed  if  one component of the state process $Y^H=(X^\pi,\Mpro^H,\Qpro^H)$, namely the filter process $\Qpro^H$ of conditional covariance is ``sufficiently small''  on the entire time interval $[0,T]$. For the regimes $H=R,J$ and for $H=Z_n$ with discrete-time expert opinions at fixed arrival times, $\Qpro^Z$ is a deterministic function  and is therefore fully specified by its initial value $\condcovinitial$. However, for the regime  $H=Z_\lambda$ with discrete-time expert opinions at random arrival times, $\Qpro^Z$ is a stochastic process that depends on the random arrival times and therefore not only specified by its initial value $\condcovinitial$.
	
	The well posedness conditions are related to non-explosive solutions of certain terminal value problems for matrix Riccati differential equations on the time interval $[0,T]$. They become more explicit for financial markets with a single risky asset. For that case the paper provides numerical results and visualizes the set of feasible model parameters.	For the actual solution of the analyzed portfolio optimization problems  we refer to  our papers \cite{Gabih et al (2022) PowerFixed,Gabih et al (2022) PowerRandom}.

	\begin{appendix}
		\section{Proof of Lemma \ref{lem:reward_condexp}}
		\label{proof_em:reward_condexp}			
		To prove the claim of this lemma for the information regimes with partial information, i.e.,  $H=R, Z_n,Z_\lambda ,J$, it is helpful to consider an auxiliary  stochastic process obtained by adding the hidden drift $\mu$ to the state process $Y^H$, and to define $\stateall^H=(Y^H,\mu)=(X^\pi,\Mpro^H,\Qpro^H,\mu)$ taking values in $\statespaceall=\mathcal{Y}^H\times \R^d=(0,\infty)\times \R^d\times \R^{d\times d} \times \R^d$. Let $\mathbb{F}^{H,\stateall}=(\mathcal{F}^{H,\stateall}_t)_{t\in[0,T]}$ with $\mathcal{F}^{H,\stateall}_t=\sigma(\stateall^H_s,~ s\le t),$ be the filtration generated by  $\stateall^H$. Then it holds
		\begin{lemma}\label{lem:Markov_augmented_state}
			For $H=R, Z_n,Z_\lambda ,J$ the process $\stateall^H$ is a Markov process, i.e., for any measurable function $g:\statespaceall\to \R$ it holds for all $s,t$ with  $0\le t\le s\le T$  
			\begin{align}\label{Markov_augmented_state}
				\E\big[g(\stateall^H_s)\,|\, \mathcal{F}_t^{H,\stateall}\big] = \E\big[g(\stateall^H_s)\,|\,\stateall^H_t\big].
			\end{align}
			
		\end{lemma}
		\begin{proof}			
			For the information regimes $H=R,J$ and the regime $H=Z_n$  the filter process $\Mpro^H$ describing the conditional mean of the hidden drift,  satisfies a SDE of the form, see \cite[Lemma 3.1--3.3]{Gabih et al (2022) PowerFixed}
			\begin{align}
				\label{Filter_M_int_gen}
				d\Mpro_t^{H}& =\alphamm(\Mpro_t^{H})\, dt
				+\betam^{H}(\Qpro_{t}^{H})\,d\widetilde{W}_t^{H},\quad \Mpro_0^{H}=\filterinitial,
			\end{align}
			driven by an $\mathbb{F}^H$-Brownian motion $\widetilde W^H$, the so-called  innovations process. For $m\in\R^d$ and $q\in\R^{d\times d}$ the drift coefficient is given by  $\alphamm(m)=\revspeed(\revlevel-m)$	and the diffusion coefficient by $ \betam^{H}(q)= q\Sigma_{\HR}^{-1/2}$, for $H=\HR, Z_n$ and by $\betam^{H}(q)=	q(\Sigma_{\HR}^{-1/2},\Sigma_{J}^{-1/2}),$ for $ H=\HD$. 			
			For $H=Z_n$ the above SDE describes the dynamics only between two arrival dates
			$t_{k-1}$ and $t_k, k=1,\ldots,n,$ whereas  at the arrival dates $t_k$  there are jumps of size $\Mpro_{t_k}^{\HC}-\Mpro_{t_k-}^{\HC} =(I_d-\rho_k)(Z_k-\Mpro_{t_k-}^{\HC})$, with the update factor	$\rho_k=\varianceexp(\Qpro_{t_k-}^{\HC}+\varianceexp)^{-1}$.
			
			For the information regimes mentioned above, the filtering process $\Qpro^H$, which describes the conditional variance of the hidden drift, is deterministic and satisfies	the Riccati differential equation, see  \cite[Lemma 3.1--3.3]{Gabih et al (2022) PowerFixed}. Here, too, the Riccati equation for $H= Z_n$ defines the dynamics only between two information times, while jumps of deterministic size occur at the information times.

			For the regime $H= Z_\lambda$ with random information dates the filter process $\Qpro^H$ is a piecewise deterministic stochastic process with jumps at the random information dates, and the filter process $\Mpro^H$ is now a jump-diffusion process. In \cite[Lemma 3.4]{Gabih et al (2022) PowerRandom} the dynamics of the two processes are given as
			\begin{align}
				d\Mpro^H_t&=\alphamm(\Mpro_t^{H})\,dt
				+\betam^H(\Qpro_t)\,d\widetilde{W}^H_t
				+\int\nolimits_{\mathbb R^{\nAktien}} \gammam(\Qpro^H_{t-},u)\; \komppoi(dt,du),\\
				%
				d\Qpro^H_t&=\alphaqq(\Qpro^H_{t})\,
				dt+\int\nolimits_{\mathbb R^{\nAktien}} \gammaq(\Qpro^H_{t-})\;	\komppoi(dt,du),	
			\end{align}
			with initial values $\Mpro^H_0=\filterinitial$,	$\Qpro^H_0=\condcovinitial$.
			As described above, the SDE for $\Mpro^H$ is driven by the innovation process $\widetilde W^H$, which is a  $\mathbb{F}^H$ Brownian motion, as well as by the compensated random measure $\komppoi$, which is associated with a homogeneous Poisson process with standard normally distributed jumps, determining the arrival times and the values of the expert opinions.
			The coefficients $\alpha$ and $\beta^H$ are as above. For details about the other coefficients ($\alphaqq,\gammam,\gammaq$) we refer to \cite[Lemma 3.4]{Gabih et al (2022) PowerRandom}.
			
			Once the dynamics of the filtering processes have been determined, the dynamics of the first component of $\stateall^H$, i.e.,  the wealth process $X^\pi$ specified in \eqref{wealth_phys} as  ${dX_t^{\pi}}/{X_t^{\pi}}=\pi_t^{\top}dR_t$, can be rewritten. To this end, the SDE \eqref{ReturnPro} for the return process $R$ is expressed in terms of the
			innovation process $\widetilde{W}$. 
			This yields the $\mathbb{F}^H$-semimartingale decomposition of $\wealth^{\pi}$, see \cite[Section 4.2]{Gabih et al (2022) PowerFixed} and \cite[Section 4.2]{Gabih et al (2022) PowerRandom}
			\begin{align}
				\frac{d\wealth_t^{\pi}}{\wealth_t^{\pi}}=\pi^{\top}_t\Mpro_t^{H}+\pi^\top_t\sigma_{\wealth}^{H}\;d\widetilde{W}_t^{H},\quad\wealth_0^{\pi}=x_0,
			\end{align}
			where $\sigma_{\wealth}^{R}=\sigma_{\wealth}^{\HC}= \Sigma_R^{1/2}$   and $\sigma_{\wealth}^{\HD}=(\Sigma_R^{1/2},~0_{\nAktien\times\nAktien})$.
			
			Finally, we recall the SDE \eqref{drift} for the hidden drift process $\mu$, which is driven by the Brownian motion $W^\mu$ and forms the fourth component of $\stateall^H$. In summary, the stochastic process $\stateall^H$ is the solution to a coupled system of SDEs driven by the Brownian motions $\widetilde W, W^\mu$, and for $H= Z_\lambda$  additionally driven by the compensated Poisson random measure $\komppoi$.    Furthermore, for the regime $H= Z_n$ with fixed information dates, the filter process $\Mpro^{Z_n}$ is updated at the these dates, and the process  $\Qpro^{Z_n}$ is deterministic. The update of $\Mpro^{Z_n}$ at time $t_k$ depends solely on its value $\Mpro^{Z_n}_{t_k-}$ before the jump, and of the random variable $Z_k$ that is independent of $\mathcal{F}^{H,\stateall}_{t_k-}$.  
			
			Summarizing,  $\stateall^H$ is a Markov process and, as such, satisfies the relation \eqref{Markov_augmented_state}.

		\end{proof}

		\begin{proof} (of Lemma \ref{lem:reward_condexp})
			We start with the proof for information regimes with partial information.
			
			\paragraph{Information regime  $\mathbf{H=R, Z_n,Z_\lambda ,J}$} We apply  Lemma \ref{lem:Markov_augmented_state} and set $s=T$ and $g(\stateall^H_T)=\utility_{\theta}(X_T^\pi)$, i.e., $g$ depends solely on the first component of $\stateall^H_T$. Then for any fixed  $t\in[0,T]$ and $\pi\in\mathcal{A}^{H}$ it holds
			\begin{align}
				\E\big[\utility_{\theta}(X_T^\pi)\,\big|\, \mathcal{F}_t^{H,X}\big] &
				= \E\big[\,\E[\utility_{\theta}(X_T^\pi)\,|\,\mathcal{F}^{H,\stateall}_t] \,\big|\,\mathcal{F}_t^{H,X}\big]\\
				&= \E\big[\,\E[\utility_{\theta}(X_T^\pi)\,|\, \stateall^H_t] \,\big|\,\mathcal{F}_t^{H,X}\big] = \E\big[\, h(\stateall^H_t) \,\big|\,\mathcal{F}_t^{H,X}\big],
				\label{hu1}
			\end{align}
			where $h:\statespaceall\to \R$ is some  measurable function that exists by  the Doob-Dynkin lemma. In the first line above, we used $\mathcal{F}_t^{H,X}\subset \mathcal{F}^\stateall_t$ and the tower property of conditional expectations, whereas in the second line we used the Markov property from \eqref{Markov_augmented_state}.
			
			We recall that $\stateall^H_t=(X^\pi_t,\Mpro^H_t,\Qpro^H_t,\mu_t)$, and that the triple $(X^\pi_t,\Mpro^H_t,\Qpro^H_t)$ is $\mathcal{F}_t^{H,X}$-measurable. Furthermore, it is known from Kalman filter theory that the conditional distribution of the hidden drift $\mu_t$ given $\mathcal{F}_t^{H,X}$ is the Gaussian distribution $\mathcal{N}(\Mpro^H_t,\Qpro^H_t)$. Therefore, there exists a $\mathcal{G}$-measurable, standard normally distributed random variable $\eps\sim \mathcal{N}(0_d,I_d)$,  that is independent of $\mathcal{F}_t^{H,X}$, such that $\mu_t=\Mpro^h_t+(\Qpro^H_t)^{1/2} \eps$. Plugging this into \eqref{hu1} we obtain 
			\begin{align}
				\E\big[\utility_{\theta}(X_T^\pi)\,\big|\, \mathcal{F}_t^{H,X}\big] & = \E\big[\, h(\stateall^H_t) \,\big|\,\mathcal{F}_t^{H,X}\big] 
				= \E\big[\, h(X^\pi_t,\Mpro^H_t,\Qpro^H_t,\mu_t) \,\big|\,\mathcal{F}_t^{H,X}\big]  	= \widetilde h(X^\pi_t,\Mpro^H_t,\Qpro^H_t)				
			\end{align}
			where the function $\widetilde  h:\mathcal{Y}^H\to\R$  is given as $\widetilde  h(y) = \widetilde h(x,m,q)= \E[\, h(x,m,q,m+q^{1/2} \eps)]  $ for $y=(x,m,q)\in\mathcal{Y}^H$. Since the above relation holds for every $t\in[0,T]$ and every $\pi\in\mathcal{A}^{H}$, we can set the function $\rewardstate^H(t,y;\pi)$ equal to $\widetilde h(y)$, which proves the claim.

			\paragraph{Information regime $\mathbf{H=F}$} In this regime we have $Y^H=Y^F= (X^\pi,\mu)$, $\mathcal {F}_t^{F}=\sigma(R_s, \mu_s,~ s\le t)$, and thus $\mathcal {F}_t^{F,X}=\sigma(R_s, \mu_s, X^\pi_s,~ s\le t)$. The state process $Y^F$ is a Markov process, since it is a solution to SDEs  \eqref{drift} and \eqref{wealth_phys}, which are driven by the Brownian motions $W^\mu$ and $W^R$. It therefore holds
			$	\E\big[\utility_{\theta}(X_T^\pi)\,\big|\, \mathcal{F}_t^{Y^F}\big] = \E\big[\utility_{\theta}(X_T^\pi)\,\big|\,Y^F_t \big].$
			Note that  SDE \eqref{wealth_phys} can be expressed without the use of the return process $R$ as ${dX_t^{\pi}}/{X_t^{\pi}}= 
			\pi_t^{\top}\mu_t\; dt+\pi_t^{\top}\volR\; dW_t^{R}$. Thus 
			the  SDEs  \eqref{drift} and \eqref{wealth_phys} defining $Y^F$ do not depend on the return process $R$,  yielding that  the pair $Y^F=(X^\pi,\mu)$ is a sufficient statistic of the triple $(X^\pi,\mu,R)$ for the inference of $\E\big[\utility_{\theta}(X_T^\pi)\,|\, \mathcal{F}_t^{F,X}\big]$, as the conditional distribution of $X^\pi_T$ given $\mathcal {F}_t^{H,X}$  does not change if we replace $\mathcal {F}_t^{F,X}$ with $\mathcal {F}_t^{Y^F}$. Thus, we obtain 	$	\E\big[\utility_{\theta}(X_T^\pi)\,\big|\, \mathcal{F}_t^{F,X}\big] = \E\big[\utility_{\theta}(X_T^\pi)\,\big|\,Y^F_t \big]$, which concludes the proof.
			
		\end{proof}

		\section{Proof of Lemma \ref{Helplemma_Psi_Expect}}
		\label{proof_lemma_Helplemma_Psi_Expect}
		\begin{proof}
			Consider first the function $\Dfun\in\mathcal C^{1,2}$ defined
			as follows
			\begin{align}
				{\Dfun}:[0,T]\times \mathbb R^{+}\times\mathbb R^{\nAktien}\rightarrow \mathbb R_{+};\qquad
				\Dfun (t,z,\filter): 
				= \E \left[\Psi_T^{t,z,\filter}\right].
				\label{Def_Psi}
			\end{align}
			Then it holds $\Dfun (T,z,\filter)= \E \left[\xi_T^{T,z,\filter}\right]=\E \left[z\right]=z$.\\
			The dynamics of the drift $\drift$ and the process $\Psi$ for $s\in[t,T]$
			read as
			\begin{align}
				\begin{pmatrix} d\mu_s^{t,\filter} \\ d\Psi_s^{t,z,\filter} \end{pmatrix}
				=
				\begin{pmatrix} \revspeed(\revlevel-\drift_s^{t,\filter}) \\ \cpsi\Psi_s^{t,z,\filter}\big(\drift_s^{t,\filter}\big)^{\top} \Sigma_{\HR}^{-1} \drift_s^{t,\filter} \end{pmatrix} ds
				+
				\begin{pmatrix} \sigma_{\mu}  \\ {0}_{1\times\nWienerDrift} \end{pmatrix}
				dW_s^{\mu}
				;\quad
				\begin{pmatrix} \drift_t^{t,\filter} \\ \Psi_t^{t,z,\filter} \end{pmatrix}
				=
				\begin{pmatrix} \filter \\ z \end{pmatrix}.
				\nonumber
			\end{align}
			The drift and the diffusion coefficients of the last equation satisfy the Lipschitz- and linear growth conditions. Moreover,
			the Feynman-Kac-Formula for the expectation from \eqref{Def_Psi}
			leads to the following partial differential equation for $\Dfun $
			\begin{align}
				0=\frac{\partial }{\partial t}\Dfun (t,z,\filter)+\nabla_{\filter}^{\top}\Dfun (t,z,\filter)\; \revspeed(\revlevel-\filter) &+\frac{1}{2}\trace\{\nabla_{\filter\filter}\Dfun (t,z,\filter) \Sigma_{\drift}\}
				\nonumber\\    &
				+ \cpsi z\filter^{\top}\Sigma_R^{-1} \filter \frac{\partial \Dfun(t,z,\filter)}{\partial z}
				\label{Cauchy_psi}
			\end{align}
			with  $\Dfun (T,z,\filter)=z$ as terminal condition  and $\nabla_{\filter}$ and $\nabla_{\filter\filter}$ denoting gradient and Hessian, respectively.  The above terminal value problem can be solved with the following separation ansatz
			\begin{align}
				\Dfun (t,z,\filter)=z\; d(t,\filter) ,\quad d(T,\filter)=1.
				\label{subst_phi_0}
			\end{align}
			At time $t$ we have $\drift_t^{t,\filter}=\filter$ and $\Psi_t^{t,1,\filter}=1$ so that we obtain 
			$\E \big[\Psi_T^{t,1,\filter}\big]=\Dfun (t,1,\filter)= d(t,\filter)$
			which is the function $d$ defined in  Lemma \eqref{Helplemma_Psi_Expect}.
			Plugging \eqref{subst_phi_0} into \eqref{Cauchy_psi} leads to the following linear parabolic PDE for $d$
			\begin{align}
				\nonumber
				0=\frac{\partial}{\partial t}d(t,\filter)+\nabla_{\filter}^{\top} d(t,\filter)\; \revspeed(\revlevel-\filter)& +\frac{1}{2}\trace\{\nabla_{\filter\filter} d(t,\filter) \Sigma_{\drift}\}+\cpsi\filter^{\top}\Sigma_R^{-1}\filter \; d(t,\filter),
			\end{align}
			with terminal value $d(T,\filter)=1$. For solving the above PDE the ansatz
			\begin{align}
				\nonumber
				d(t,\filter)=\exp\big\{\filter^{\top}\Abound(t)\filter+\Bbound^{\top}(t)\filter+\Cbound(t)\big\}
			\end{align}
			leads to the system of ODEs  for  $\Abound$, $\Bbound$ and
			$\Cbound$, which are given in \eqref{A_bound}, \eqref{B_bound} and \eqref{C_bound}.  \myqed
		\end{proof}
		\section{Proof of Theorem \ref{theo_bound_V}}
		\label{proof_lemma_lemma_bound_V}
		\begin{proof}
			In the proof we follow an approach presented in  Angoshtari \cite[Theorem 1.8]{Angoshtari2013}.
			Without loss of generality we give the proof for $t=0$ and show that it holds  $V_\para ^G(0,y)\leq \frac{x^{\theta}}{\theta} \,
			d^{1-\theta}( 0,m)$ for all $y=(x,m)\in \mathcal{Y}^G=(0,\infty)\times \R^d$.

			Let $(\xi_t)_{t\in[0,T]}$ be a stochastic process satisfying the SDE
			\begin{align}
				d\xi_t= -\xi_t \drift_t^{\top}\;\Sigma_{\HR}^{-1}\volR\;dW_t^{\HR} ,\quad \xi_0=1, \; ~ \drift_0=m ,
				\label{xi_prozess_DGL}
			\end{align}
			with the solution
			\begin{align}
				\xi_t=\exp\Big\{ -\frac{1}{2}\int\nolimits_0^t \| \drift_s^{\top}\;\Sigma_{\HR}^{-1}\volR \|^2 ds -\int\nolimits_0^t  \drift_s^{\top}\;\Sigma_{\HR}^{-1}\volR\;dW_s^{\HR}\Big\}.
				\label{xi_prozess}
			\end{align}
			For $t_0\in[0,T]$ we denote by $\drift_t^{t_0,{m} }$ the solution to the  SDE \eqref{drift} for the drift process $\drift$
			starting at time $t_0$ with initial value ${m} $, by $\welth_t^{\pi,t_0,x,{m} }$ the solution to the wealth equation \eqref{wealth_phys}
			with initial values $(x,{m} )$ and by $\xi_t^{t_0,z,{m} }$ the solution of \eqref{xi_prozess_DGL}
			at time $t$ with initial values $(z,{m} )$.
			Applying It\^{o}'s-formula it holds
			\begin{align}
				d(\welth^{\pi,0,x,{m}  }~\xi^{0,1,{m}  })_t&=\welth_t^{\pi,0,x,{m}  }~d\xi_t^{0,1,{m}  }+\xi_t^{0,1,{m}  }\;d\welth_t^{\pi,0,x,{m}  }+d\langle \welth^{\pi,0,x,{m}  },\xi^{0,1,{m}  }\rangle_t\nonumber\\
				&=\xi_t^{0,1,{m}  }~\welth_t^{\pi,0,x,{m}  }[\pi_t^{\top}\volR-\drift_t^{0,{m}  }\Sigma_{\HR}^{-1}\volR]\;dW_t^{\HR}.
			\end{align}
			Moreover, Fatou's Lemma implies that the non-negative process $(\welth^{\pi}\xi)_t$ is a
			supermartingal, and as a consequence it holds
			\begin{align}
				x-\E [\welth_T^{\pi,0,x,{m}  }~\xi_T^{0,1,{m}  } ]\geq0.
				\label{super_M1}
			\end{align}
			From the other hand let $f:\mathbb R^+\rightarrow \mathbb R$ be the associated Legendre-Fenchel transformation of the utility function
			$ \utility_{\theta}(x)$
			defined for every $w>0$ by
			\begin{align}
				\label{Konjugierte}
				f(w):=\sup\limits_{x\in\mathbb R^+}\left\{ \utility_{\theta}(x)-xw\right\}=
				\frac{1-\theta}{\theta}w^{-\frac{\theta}{1-\theta}}.
			\end{align}
			Since $\xi_T^{0,1,{m}  }>0$, it holds for every $w>0$
			\begin{align}
				\label{Konj_Disk}
				f(\xi_T^{0,1,{m}  }~ w) =\sup\limits_{x\in\mathbb R^+} \Big\{\utility_{\theta}(x)-x\;\xi_T^{0,1,{m}  }\; w\Big\} \geq \utility_{\theta}(\welth_T^{\pi,0,x,{m}  })-\welth_T^{\pi,0,x,{m}  }\; \xi_T^{0,1,{m}  }\; w.
			\end{align}
			Now for $w>0$ and  $y=(x,m)\in \mathcal{Y}^G=(0,\infty)\times \R^d$  inequality \eqref{super_M1} implies that
			\begin{align}
				\rewardstate_\para ^G(0,y;\pi) = \E \big[\utility_{\theta}(\welth_T^{\pi,0,x,{m}  })\big]&
				\leq \E \big[ \utility_{\theta}(\welth_T^{\pi,0,x,{m}  }) \big]+w \big(x-\E \big[\welth_T^{\pi,0,x,{m}  }~\xi_T^{0,1,{m}  }\big]\big)\nonumber\\
				&=\E \big[\utility_{\theta}(\welth_T^{\pi,0,x,{m}  })-\welth_T^{\pi,0,x,{m}  }\;\xi_T^{0,1,{m}  }\; \text{$w$}\big]+x w\nonumber\\
				& \leq \E \Big[f(\xi_T^{0,1,{m}  }~w)\Big]+x \text{$w$},\nonumber
			\end{align}
			where the last inequality follows from \eqref{Konj_Disk}.
			For the term $f(\xi_T^{0,1,{m}  }~w)$  we now apply
			\eqref{Konjugierte} to obtain
			\begin{align}
				\rewardstate_\para ^G(0,y;\pi)  
				&	\leq\frac{1-\theta}{\theta} w^{-\frac{\theta}{1-\theta}}
				\E \Big[\big(\xi_T^{0,1,{m}  }\big)^{-\frac{\theta}{1-\theta}}\Big]+xw.
				\label{Zielfkt_Absch2}
			\end{align}
			Since the last inequality holds for every admissible strategy $\pi\in\mathcal A^{\HG}$ and for every $w>0$,
			we can take the supremum over all strategies $\pi\in\mathcal
			A^{\HG}$ on the left-hand side and the infimum  over all $w>0$ in the
			right-hand side to obtain
			\begin{align}
				V_\para ^G(0,y)   =\sup\limits_{\pi \in\mathcal A^{\HG}} \rewardstate_\para ^G(0,y;\pi)
				&\leq   \frac{x^{\theta}}{\theta}\
				\Big( \E \big[\big(\xi_T^{0,1,{m}  }\big)^{-\frac{\theta}{1-\theta}}\big] \Big)^{1-\theta}.
				\label{Zielfkt_Absch3}
			\end{align}
			The problem is now reduced to investigate if  the expectation in the r.h.s. of \eqref{Zielfkt_Absch3}
			is bounded. It holds
			\begin{align}
				\big(\xi_T^{0,1,\filter}\big)^{-\frac{\theta}{1-\theta}}&=\exp\Big\{ \frac{\theta}{1-\theta}\Big(
				\frac{1}{2}\int\nolimits_0^T \big\| \big(\drift_s^{0,\filter}\big)^{\top}\Sigma_{\HR}^{-1}\volR \big\|^2\;ds
				+\int\nolimits_0^T\big(\drift_s^{0,\filter}\big)^{\top}\Sigma_{\HR}^{-1}\volR \; dW_s^{\HR} \Big)\Big\}
				\\&
				=\Lambda_T \cdot \Psi_T^{0,1,\filter},				
			\end{align}
			where $\Psi_T^{0,1,\filter}$ is given in \eqref{Psi_pro} with $\cpsi=\frac{\theta}{2(1-\theta)^2}$. The term $\Lambda_T$
			is given by
			\begin{align}
				\Lambda_T=
				\exp\Big\{\int\nolimits_0^T \frac{\theta}{1-\theta} \big(\drift_s^{0,\filter}\big)^{\top}\Sigma_{\HR}^{-1}\volR \; dW_s^{\HR}
				-\frac{1}{2}\int\nolimits_0^T \big\|\frac{\theta}{1-\theta} \big(\drift_s^{0,\filter}\big)^{\top}\Sigma_{\HR}^{-1}\volR \big\|^2\;ds\bigg\}.
				\nonumber
			\end{align}
			We now introduce a new probability measure ~$\P^{\ast}$ given by $\frac{d\P^{\ast}}{d\P~}=\Lambda_T$
			so that the expectation from \eqref{Zielfkt_Absch3} can be expressed as
			\begin{align}
				\E \Big[\big(\xi_T^{0,1,\filter}\big)^{-\frac{\theta}{1-\theta}} \big]&
				=\E \big[\Lambda_T \cdot \Psi_T^{0,1,\filter}\big]
				=\E ^{\ast}\Big[\Psi_T^{0,1,\filter}\Big]=d(0,\filter).
				\nonumber
			\end{align}
			This expectation can be expressed according to \eqref{Psi_4} in
			Lemma \ref{Helplemma_Psi_Expect}  and its proof given in Appendix \ref{proof_lemma_Helplemma_Psi_Expect},  where  $\E ^{\ast}$ denotes the expectation  under the new probability measure.
			\myqed
		\end{proof}		
		\section{Proof of Lemma \ref{moment_quadratic_form}}
		\label{proof_moment_quadratic_form}
		\begin{proof}		
			\paragraph{First claim}
			Note that the product of the symmetric matrices $\Sigma$ and $U$ needs not to be symmetric, the latter would immediately imply real eigenvalues. 
			
			Since $\Sigma$ is positive semidefinite there exists a $d\times d$-matrix $\Sigmafac$ such that   $\Sigma=\Sigmafac\Sigmafac^{\top}$. It is well-known that for $d\times d$ matrices $A,B$ it holds that $AB$ and $BA$ have the same eigenvalues.
			Setting $A=P$ and $B=P^\top U$ it follows that $\Sigma U=PP^\top U$ and $P^\top U P$ have the same eigenvalues. They are real since $P^\top U P$ is symmetric. 
			
			\paragraph{Second claim}			
			The decomposition    $\Sigma=\Sigmafac\Sigmafac^{\top}$ allows the representation $Y=\Sigmafac Z$		
			with an $d$-dimensional standard normally distributed random vector $Z=\mathcal N(0_d,I_d)$, since the mean of $\Sigmafac Z$ is $\E[\Sigmafac Z]=0_d$ and its covariance matrix is $\E[\Sigmafac ZZ^\top\Sigmafac^\top]=PI_dP^\top=\Sigma$. Then, we have
			$Y^{\top}{U}Y +a^{\top}Y=Z^{\top}{\Sigmafac}^{\top}{U}{\Sigmafac}Z+a^{\top}\Sigmafac Z,		$
			so that it holds
			\begin{align}
				\E[\,e ^\quadform\,]=\E[\exp{(Y^{\top}{U}Y +a^{\top}Y)}]&=\E[\exp{(Z^{\top}{\Sigmafac}^{\top}{U}{\Sigmafac}Z+a^{\top}\Sigmafac Z})]\\
				&=\frac{1}{(2\pi)^{d/2}}\int_{\R^d}\exp({z^{\top}{\Sigmafac}^{\top}{U}\Sigmafac z+a^{\top}{\Sigmafac}z-\frac{1}{2}z^{\top}z})dz.
				\label{integrand}
			\end{align} 
			The eigenvalues of $\Sigma U$ are real and satisfy by assumption  $\lambda_i<\frac{1}{2}, i=1,\ldots,n$. This implies that the eigenvalues of  $K=I_d-2\Sigma{U}$ which are   given by $1-2\lambda_i$ are real and  positive, and $K$ is invertible.
			We define the $d\times d$ matrix $G=I_d-2{\Sigmafac}^{\top}{U}\Sigmafac$. Recall that the matrices $\Sigma U$ and $\Sigmafac^{\top}U\Sigmafac$ have the same  eigenvalues. Thus  $G$ is also invertible since its eigenvalues are given by $1-2\lambda_i>0$. Further, it holds 
			\begin{align}
				\label{det_identity}
				\det(G)  =\det(K)= \prod_{i=1}^d(1-2\lambda_i).   
			\end{align}      
			Rearranging terms in the integral  of Equation \eqref{integrand} yields
			\begin{align}
				z^{\top}{\Sigmafac}^{\top}{U}\Sigmafac z+a^{\top}{\Sigmafac}z-\frac{1}{2}z^{\top}z
				&=-\frac{1}{2}(z^{\top}Gz-2z^{\top}\Sigmafac^{\top}a)=\frac{1}{2}a^{\top}\Sigmafac G^{-1}{\Sigmafac}^{\top}a-\frac{1}{2}(z-b)^{\top}G(z-b),
			\end{align} 
			with $b=G^{-1}{\Sigmafac}^{\top}a$. 
			Using that $h(z)=(2\pi)^{- d/2} (\det(G^{-1}))^{-1/2}\exp\{-\frac{1}{2}(z-b)^{\top}G(z-b\}$ is the probability density function of the non-degenerate Gaussian distribution $\mathcal{N}(b,G^{-1})$ with the normalization $\int_\R h(z)\,dz=1$, it follows from \eqref{integrand}	and \eqref{det_identity}	
			\begin{align}
				\E[\,e^\quadform\,]
				&=(\det(G))^{-1/2}\exp\Big\{\frac{1}{2}a^{\top}\Sigmafac G^{-1}{\Sigmafac}^{\top}a\Big\}
				=(\det(K))^{-1/2}\exp\Big\{\frac{1}{2}a^{\top} K^{-1}\Sigma a\Big\}.
			\end{align} 			
			For the last equality we have applied the fact that $\Sigmafac G^{-1}{\Sigmafac}^{\top}=K^{-1}\Sigma$. The proof of this equality is based on the fact that for $d\times d$-matrices $A$ and $B$ which are such that  $C=I_d-AB$ is invertible,  it holds 
			$$(I_d-AB)^{-1}=I_d+A(I_d-BA)^{-1}B.$$ 
			This can easily seen by verifying the defining property of an inverse matrix, i.e., $CC^{-1}=C^{-1}C=I_d$.
			Setting  $A=\Sigmafac^{\top}U$ and $B=\Sigmafac$ we obtain the identity $\Sigmafac G^{-1}{\Sigmafac}^{\top}=K^{-1}\Sigma$ from
			\begin{align}
				G^{-1}=(I_d-2P^\top UP)^{-1} = I_d+2\Sigmafac^{\top}U(I_d+2PP^\top U)^{-1} \Sigmafac = I_d+2\Sigmafac^{\top}UK^{-1} \Sigmafac
			\end{align}
			and finally 
			\begin{align}
				\Sigmafac G^{-1}{\Sigmafac}^{\top}&=\Sigmafac(I_d+2\Sigmafac^{\top}UK^{-1} \Sigmafac)\Sigmafac^{\top}=\Sigmafac \Sigmafac^{\top}+2\Sigmafac \Sigmafac^{\top}UK^{-1}\Sigmafac \Sigmafac^{\top}=\Sigma+2\Sigma UK^{-1}\Sigma\nonumber\\
				&=(I_d+2\Sigma UK^{-1})\Sigma
				=(K+2\Sigma U)K^{-1}\Sigma=I_dK^{-1}\Sigma.
				\label{matrix_equality}
			\end{align}
			
			\paragraph{Third claim} The first identity in \eqref{quadratic form_eigenvalue} was already proven in \eqref{det_identity}. For the second identity we use $\Sigmafac^{\top}U\Sigmafac= D\Lambda D^\top$ and that  the matrix D is  orthogonal, i.e., $DD^\top=D^\top D=I_d$. Then, according to \eqref{matrix_equality} it holds
			\begin{align*}
				K^{-1}\Sigma&=\Sigmafac G^{-1}{\Sigmafac}^{\top}=\Sigmafac (I_d-2\Sigmafac^{\top} U\Sigmafac)^{-1}\Sigmafac^{\top}=\Sigmafac (DD^{\top}-2DD^{\top}\Sigmafac^{\top} U\Sigmafac DD^{\top})^{-1}\Sigmafac^{\top}\\
				&=\Sigmafac \big(D(I_d-2D^{\top}\Sigmafac^{\top} U\Sigmafac D)D^{\top}\big)^{-1}\Sigmafac^{\top}=\Sigmafac \big(D(I_d -2\Lambda) D^{\top}\big)^{-1}\Sigmafac^{\top}
				=\Sigmafac D(I_d -2\Lambda)^{-1}D^{\top}\Sigmafac^{\top}.
			\end{align*}
			Using $c=D^{\top}\Sigmafac^{\top}a$  we obtain
			\begin{align}
				a^{\top}K^{-1}{\Sigma} a= a^\top\Sigmafac D(I_d -2\Lambda)^{-1}D^{\top}\Sigmafac^{\top}a 
				= c^\top(I_d -2\Lambda)^{-1} c=\sum\limits_{j=1}^d c_j^2(1-2\lambda_j)^{-1}.			
			\end{align}
			\myqed

		\end{proof}

		\section{Proof of Theorem \ref{theorem_partial_Inv}}
		\label{proof_theorem_partial_Inv}
			\begin{proof}
				
				We recall inequality \eqref{bound_value_H} stating 		
				$V_\para^{H}(t,y)\leq
				\frac{x^{\theta}}{\theta}\E \Big[d^{  1-\theta}(t,\drift_t) \big| \Mpro^H_t=m,\Qpro_t^H=q\Big]$, for $H=R, Z_n,Z_\lambda ,J$, and $y=(x,m,q)$.
				For the $H$-investors  the  conditional distribution of $\mu_t$ 
				given $\Mpro^H_t=m,\Qpro_t^H=q$ is the Gaussian distribution $\mathcal N(m,q)$.  Thus we can deduce for the conditional expectation  $$\E [ d^{  1-\theta}(t,\drift_t)  | \Mpro^H_t=m,\Qpro_t^H=q]=  \E \big[ d^{  1-\theta}(t,m+q^{1/2}\varepsilon )  \big]$$
				with a random variable $~\varepsilon\sim\mathcal N(0,I_{\nAktien})$ independent of $\mathcal F_t^{H}$. Substituting into \eqref{bound_value_H}  and using representation \eqref{Psi_4}  we deduce 
				\begin{align}
					V_\para^{H}(t,y)&\leq\E [ d^{ 1-\theta}(t,\drift_t)  | \Mpro_t^H=m,\Qpro_t^H=q] =
					\E \big[ d^{ 1-\theta}(t,m+q^{1/2}\varepsilon )  \big]\nonumber\\
					&=\E \big[\exp\big\{{ (1-\theta)} \big((m+{q}^{1/2}\varepsilon)^{\top} \Abound(t) (m+{q}^{1/2}\varepsilon) +\Bbound^{\top}(t) (m+{q}^{1/2}\varepsilon)  +\Cbound(t)\big)\big\} \big].
				\end{align}
				To simplify the notation we write in the following $A,B,C$ instead of ${ (1-\theta)}\Abound(t)$, ${ (1-\theta)}\Bbound(t)$, ${ (1-\theta)}\Cbound(t)$,  respectively. Rearranging terms yields
				\begin{align}
					V_\para^{H}(t,y)&\le \E \big[\exp\big\{ m^\top A m+B^{\top} m +C\big\} \exp\big\{
					({q}^{1/2}\varepsilon)^{\top} A{q}^{1/2}\varepsilon +(2m^\top A+B^{\top}) {q}^{1/2}\varepsilon \big\} \big]\nonumber\\
					&=d^{ 1-\theta}(t,m)\E \big[ \exp\big\{
					({q}^{1/2}\varepsilon)^{\top} A{q}^{1/2}\varepsilon +(2 Am+B)^{\top} {q}^{1/2}\varepsilon \big\} \big]\nonumber\\					
					&=	{d^{ 1-\theta}(t,m)  \E \big[\exp\{Y^{\top}{A}Y +a^\top Y\}\big]},
					\label{exponents}
				\end{align}
				where $Y=q^{1/2}\varepsilon \sim\mathcal{N}(0_d,\Sigma)$
				is a zero-mean Gaussian random vector with 		covariance matrix $\Sigma=q$ and $a=2Am+B ={ (1-\theta)(2\Abound(t)m+\Bbound(t))}$.    Applying Lemma \ref{moment_quadratic_form} with ${U}=A={ (1-\theta)}\Abound(t)$ and ${K}=I_{\nAktien}-2q A = { I_{\nAktien}- 2(1-\theta)q }\Abound(t)$ 
				yields  
				\begin{align*}
					\E \big[\exp\{Y^{\top}{A}Y\} +a^\top Y\big] &=
					\big(\rm{det}(K)\big)^{-1/2}
					\exp\big\{a^{\top}\;   K^{-1} q \,a\big\},
				\end{align*}
				and substituting this expression into  \eqref{exponents} proves the claim. 
				\myqed
				
			\end{proof}

		\end{appendix}
		
		\smallskip\noindent\textbf{Acknowledgments~}
		The authors thank    Jörn Sass   (RPTU University  Kaiserslautern-Landau)	for valuable discussions that improved this paper.

		\let\oldbibliography\thebibliography
		\renewcommand{\thebibliography}[1]{%
			\oldbibliography{#1}%
			\setlength{\itemsep}{-.85ex plus .05ex}  
		}
		\bibliographystyle{amsplain}

	\end{document}